\documentclass[11pt,reqno]{amsart}
\usepackage[T1]{fontenc}  
\usepackage[utf8]{inputenc}  
\usepackage[english]{babel} 

\usepackage[a4paper,hmargin=3cm,vmargin=3cm]{geometry}

\usepackage{xcolor}
\RequirePackage[colorlinks=True,linkcolor=blue,urlcolor=blue]{hyperref} 
\RequirePackage{graphicx} 
\usepackage{mathrsfs} 
\usepackage[font=scriptsize,labelfont={scriptsize}]{subfig} 
\usepackage{bbm} 
\usepackage{enumerate}
\usepackage{enumitem,kantlipsum}
\usepackage{caption} 
\usepackage{comment}
\usepackage{cite}
\usepackage{amssymb, amsmath}
\usepackage{nicefrac}
\usepackage{dsfont}

\numberwithin{equation}{section}
\setlist[enumerate]{leftmargin=*,label=(\roman*)}

\theoremstyle{plain}
\newtheorem{theorem}{Theorem}[section]

\newtheorem{lemma}[theorem]{Lemma}
\newtheorem{corollary}[theorem]{Corollary}

\theoremstyle{remark}
\newtheorem{definition}[theorem]{Definition}
\newtheorem{example}[theorem]{Example}

\newtheorem{assumption}[theorem]{Assumption}

\newtheorem{remark}[theorem]{Remark}

\newcommand{\eps}{\varepsilon}
\renewcommand{\phi}{\varphi}

\newcommand{\N}{\mathbb{N}}

\newcommand{\Q}{\mathbb{Q}}
\newcommand{\R}{\mathbb{R}}

\renewcommand{\P}{\mathbb{P}}
\newcommand{\E}{\mathbb{E}}

\newcommand{\cB}{\mathcal{B}}

\newcommand{\cE}{\mathcal{E}}

\newcommand{\cH}{\mathcal{H}}

\newcommand{\cL}{\mathcal{L}}

\newcommand{\cN}{\mathcal{N}}

\newcommand{\cQ}{\mathcal{Q}}

\newcommand{\cT}{\mathcal{T}}

\newcommand{\cW}{\mathcal{W}}
\newcommand{\cX}{\mathcal{X}}

\renewcommand{\d}{{\rm d}}

\newcommand{\Lip}{{\rm Lip}}

\newcommand{\Rd}{{\R^d}}
\def\BUC{\rm BUC}

\def\Cb{{\rm C}_{\rm b}}
\def\Cbi{{\rm C}_{\rm b}^\infty}

\def\Lipb{{\rm Lip}_{\rm b}}
\def\one{\mathds{1}}
\def\LI{\cL_I}
\def\LS{\cL_S}

\DeclareMathOperator{\id}{id}
\DeclareMathOperator{\Tr}{Tr}

\begin{document}

\title[Risk-based prices]{Discrete approximation of risk-based prices under volatility uncertainty}

\author{Jonas Blessing}
\address{ETH Zurich, Department of Mathematics, R\"amistra\ss{}e 101, 8092 Zurich, Switzerland}
\email{jonas.blessing@math.ethz.ch}

\author{Michael Kupper}
\address{Department of Mathematics and Statistics, University of Konstanz.}
\email{kupper@uni-konstanz.de}

\author{Alessandro Sgarabottolo}
\address{Center for Mathematical Economics, Bielefeld University}
\email{alessandro.sgarabottolo@uni-bielefeld.de}

\date{\today}
\thanks{Financial support through the Deutsche Forschungsgemeinschaft (DFG, German Research Foundation) -- SFB 1283/2 2021 -- 317210226 is gratefully acknowledged by the third named author.}

\begin{abstract}
We discuss the asymptotic behaviour of risk-based indifference prices of European contingent 
claims in discrete-time financial markets under volatility uncertainty as the number of 
intermediate trading periods tends to infinity.\ The asymptotic risk-based prices form a 
strongly continuous convex monotone semigroup which is uniquely determined by its 
infinitesimal generator and therefore only depends on the covariance of the random factors
but not on the particular choice of the model.\ We further compare the risk-based prices
with the worst-case prices given by the $G$-expectation and investigate their asymptotic 
behaviour as the risk aversion of the agent tends to infinity.\ The theoretical results are 
illustrated with several examples and numerical simulations showing, in particular, that 
the risk-based prices lead to a significant reduction of the bid-ask spread compared to 
the worst-case prices. 

\smallskip
		
\end{abstract}

\maketitle

\section{Introduction}

Computing the prices of financial derivatives strongly depends on the choice of the underlying
model and the associated probability distributions.\ Since these distributions are, in general, 
not precisely known, robust finance takes into account model uncertainty by considering sets of 
possible transition probabilities.\ In this article, we start with a simple asset model in 
discrete time for which derivative prices can easily be computed by backward recursion and analyze 
the asymptotic behaviour of derivative prices as the number of intermediate trading periods tends 
to infinity.\ A classical example of this is the convergence of derivative prices in the binomial 
model to the Bachelier or Black--Scholes prices, see e.g.~\cite[Section~5.7]{foellmer2016finance}.\ 
Discrete financial models are generally straightforward from a modeling perspective and arise
naturally since trading typically occurs at discrete time points. Nonetheless, continuous-time 
models are very popular since they allow for the use of stochastic calculus and PDE methods.\ 
Furthermore, it has recently been shown in~\cite{Criens2024} that superhedging prices of 
discrete-time models with uncertain Markovian transition kernels converge to superhedging prices 
of continuous-time models with drift and volatility uncertainty. Superhedging prices correspond 
to intervals of plausible prices which do not generate arbitrage opportunities, see for example \cite{delbaen2006arbitrage}. 
Although such intervals can naturally be associated with an arbitrage-free bid-ask spread, 
these bounds are, in general, too wide to be informative about the prices of a contingent claim 
in incomplete markets.\ In fact, each agent operating in the market assigns a different subjective
value to the same contingent claim which can violate the bounds prescribed by the superhedging 
or subhedging prices.
    
A classical approach to reduce the bid-ask spread observable in incomplete markets consists of taking 
the preferences of an agent into account by associating a utility function or a risk 
measure to the agent.\ A strand of literature has developed in this direction introducing 
so-called good deals bounds.\ Good deals bounds aim to reduce the no-arbitrage bounds by ruling 
out prices that can be hedged by strategies leading to a high expected utility, i.e., prices 
that represent a deal which is too good.\ Good deals have been measured by means of Sharpe 
ratios~\cite{Cochrane2000Beyond}, gain-loss ratios~\cite{Bernardo2000Gain} or utility 
functions~\cite{Carr2001incomplete,Bjork2006General,Cerny2002good}. Moreover, the setting 
in~\cite{Foellmer1999Quantile} allows for imperfect hedges as long as their level of risk is 
acceptable.\ Closely related to good deal prices are indifference prices which make the agent 
indifferent regarding her utility or risk between selling or keeping the derivative, 
see~\cite{Xu2006Risk,Schied2006Risk,Kloppel2007Indifference,Biagini2011Indifference,BionNadal2020FullyDynamic}.\ 
A connection between the two concepts was first established in~\cite{Jaschke2001Coherent}, 
where it is shown that indifference prices based on coherent risk measures are equivalent 
to the good deal prices in~\cite{Cerny2002good}.\ Furthermore, every convex risk measure 
representing good deal prices is given by an indifference price, see~\cite{Arai2014gooddeal}.\ 
An extensive collection of the literature on good deal and indifference prices can be found 
in~\cite{Carmona2009Indifference}.\ So far, explicit solutions have only been given in dominated 
settings, where the asset process is given with respect to a physical measure $\P$ and the absence 
of good deals translates into restrictions on the set of equivalent local martingale measures 
$\{\Q\}_{\Q \sim \P}$.\ However, since incompleteness in a market is naturally connected with model 
uncertainty and the inability to precisely estimate the distribution of the assets, a more general 
framework taking model uncertainty into account seems to be necessary.\ In addition, as pointed out 
in~\cite{Staum2004Fundamental}, classical indifference pricing often has the flavour of a one-period model.\

In this article, we work in a non-dominated setting, where the uncertain distribution of
the increments of the asset process is determined by a sublinear expectation $\cE[\,\cdot\,]$.
Moreover, we consider an agent who measures the risk associated to a random loss $Y$ by means 
of a robust entropic risk measure 
\[ \rho[Y]:=\frac{1}{\alpha}\log\big(\cE\big[e^{\alpha Y}\big]\big), \] 
where $\alpha>0$ is a risk aversion parameter.\ Since the entropic risk measure is the certainty 
equivalent of the exponential utility, our agent can also be seen as a robust exponential utility 
maximizer whose preferences belong to the class of multiple priors preferences which have been 
introduced and characterized in~\cite{Gilboa1989Maxmin}. The term \textit{robust} refers to the 
consideration of several plausible models which can be derived from the dual representation 
\[ \cE[Y]=\sup_{\Q\in\cQ}\E_\Q[Y] \]
of the sublinear expectation.\ For a brief introduction to sublinear expectations, we refer to 
Subsection~\ref{sec:expectation} and the references therein.\ Following the previously mentioned 
work~\cite{Gilboa1989Maxmin} and its extension~\cite{Maccheroni2006Ambiguity}, the problem of 
robust pricing has gained a great deal of attention. In particular, we refer 
to~\cite{beiglboeck2013model,Dolinsky2014Martingale,Bouchard2015Arbitrage,acciaio2016modelfree,
Beiglboeck2017DualityMOT,Burzoni2019Pointwise} for arbitrage theory and superhedging dualities 
under model uncertainty and to~\cite{Dolinsky2014Robust,Cheridito2017Duality,Dolinsky2017Convex} 
for similar results in the presence of transaction costs and trading constraints.\ The multiple 
priors problem has also been tackled in the context of utility maximization in dominated settings,
see~\cite{Quenez2004Optimal,Gundel2005Robust,Owari2011Robust,Backhoff2016Robust} and in
non-dominated settings, see~\cite{Denis2013Optimal,Nutz2016Utility,Neufeld2018Robust,Neufeld2019Nonconcave}.\
In the specific context of exponential utility, the authors of~\cite{Frittelli2000Minimal,Delbaen2002Exponential} 
prove duality in a continuous-time non-robust setting.\ Furthermore, in~\cite{Rouge2000Pricing,Hu2005Utility} 
the value function of the utility maximization problem is characterized via a quadratic BSDE 
and several properties of the pricing functional such as monotonicity and the asymptotic behaviour 
w.r.t. the risk aversion parameter are derived.\ These results have been extended in~\cite{Mania2005Dynamic} 
and, for a non-dominated setting, in~\cite{Matoussi2015Robust}. For unbounded claims, duality 
and the existence of maximizers have been established in~\cite{Bartl2019Exponential} and, 
under the presence of transaction costs, in~\cite{Deng2020Utility}. 
    
Starting from a $d$-dimensional discrete-time market, we assume that the asset process~$X$ 
has independent increments which are determined by the equation
\[ X_{(k+1)h}-X_{kh}=h\mu+\sqrt{h}\zeta \quad\mbox{for all } k\in\N, \]
where $h>0$ is a fixed step-size, $\mu\in\Rd$ is a deterministic drift and $\zeta$ is a 
$d$-dimensional random vector with mean zero. Given a claim with payoff function $f$ and 
maturity $T=Nh$, the indifference ask price $a_T(f)\in\R$ is uniquely determined by the relation
\[ \inf_{\theta=(\theta_1,\dots,\theta_N)}\rho\big[f(X_T)-a_T(f)-(\theta\cdot X)_T\big]
    =\inf_{\theta=(\theta_1,\dots,\theta_N)}\rho[-(\theta\cdot X)_T], \]
where the random variables $\theta_1,\dots,\theta_N\colon\Omega\to\Theta$ take values in a set 
of available strategies $\Theta\subset\Rd$ such that $\theta_k$ is $X_{(k-1)h}$-measurable for 
all $k=1,\ldots,N$ and 
\[ (\theta\cdot X)_T:=\sum_{k=1}^N\theta_k(X_{kh}-X_{(k-1)h}). \]
Under suitable conditions, one can show that $a_T\colon\Cb\to\Cb$, where $\Cb$ consists of all 
bounded continuous functions $f\colon\Rd\to\R$. Hence, if we require the indifference ask prices
to be time consistent, they are completely determined by the one-step pricing operator $I(h)f:=a_h(f)$
and the equation
\begin{equation} \label{eq:intro1}
 a_T(f)=I(h)^N f \quad\mbox{for all } T=Nh \mbox{ and } f\in\Cb.
\end{equation}
Dynamic consistency has previously been imposed in~\cite{Cochrane2000Beyond} to solve the 
multi-period pricing problem by iterating the solution of the one-period model. Furthermore, in a 
continuous-time setting, the authors of~\cite{Kloppel2007Dynamic} show that \textit{local} 
conditions on the pricing kernels guarantee \textit{nice} global properties of the pricing operator, 
included time-consistency. In Subsection~\ref{sec:dyn.consist}, dynamically consistent pricing 
operators are discussed in more detail.\ We are now interested in the limit behaviour of the multi-step
prices as the number of intermediate trading periods tends to infinity. For that purpose, 
let $t\geq 0$ be a maturity and $h_n:=t/n$ be a sequence of decreasing step-sizes. Then, if the limit 
\begin{equation} \label{eq:intro2}
 a_t^\infty(f):=\lim_{n\to\infty}a_t^{(n)}(f)
\end{equation}
of the multi-step indifference prices exists, it defines the asymptotic risk-based price of a claim 
with payoff function $f$ and maturity $t$. In order to prove that the previous limit exists
and to uniquely characterize the global dynamics of the asymptotic risk-based prices by means 
of their infinitesimal bevahiour, we reformulate equation~\eqref{eq:intro2} as an approximation 
result of a nonlinear semigroup.\ This view point is motivated by the fact that the time consistency 
of the multi-period prices transfers to the limit, i.e.,
\begin{equation} \label{eq:intro3}
 a_{s+t}^\infty(f)=a_s^\infty(a_t^\infty(f)) \quad\mbox{for all } s,t\geq 0.
\end{equation}
Equation~\eqref{eq:intro1} guarantees that the right-hand side of equation~\eqref{eq:intro2}
is given by
\[ a_t^{(n)}(f)=I\big(\tfrac{t}{n}\big)^n f 
    \quad\mbox{for all } t\geq 0, f\in\Cb \mbox{ and } n\in\N, \]
where $I(h)\colon\Cb\to\Cb$ denotes the pricing operator for one period with step size $h\geq 0$. 
Furthermore, these operators have desirable properties such as convexity and monotonicity. 
The question whether a sequence of iterated
operators $(I(\tfrac{t}{n})^n)_{n\in\N}$ on $\Cb$ converges to a limit operator $S(t)\colon\Cb\to\Cb$ 
such that $(S(t))_{t\geq 0}$ is a strongly continuous convex monotone semigroup has systematically
been addressed in a series of recent articles, see~\cite{BDKN22,BK22,BK23,BKN23}.
Applying these results to the present setting shows that the limit
\[ S(t)f:=\lim_{n\to\infty}I\big(\tfrac{t}{n}\big)^n f \]
exists for all $t\geq 0$ and $f\in\Cb$.\ Moreover, the family $(S(t))_{t\geq 0}$ is a strongly 
continuous convex monotone semigroup which is uniquely determined by its infinitesimal generator 
\[ Af:=\lim_{h\downarrow 0}\frac{S(h)f-f}{h}. \]
The asymptotic risk-based prices are then defined by 
\[ a_t^\infty(f):=S(t)f \quad\mbox{for all } t\geq 0 \mbox{ and } f\in\Cb. \]
Note that, by equation~\eqref{eq:intro2} and~\eqref{eq:intro3}, these prices are time consistent 
and given as the limit of multi-period prices in a discrete model.\ In addition, the prices
have desirable properties such as convexity and monotonicity w.r.t. the payoff function.\
So far, we did not address the question whether the asymptotic risk-based prices depend on
the particular choice of the discrete model. In case of the previously mentioned approximation 
of the Bachelier prices, the central limit theorem guarantees that the Bachelier prices only
depend on the covariance of the discrete model but not on the particular choice of its
distribution.\ This observation is a particular case of the more general statement that 
strongly continuous semigroups are uniquely determined by their generators.\ For linear 
semigroups the uniqueness is classical result and for convex monotone semigroups it has 
been proven in~\cite{BDKN22,BK23}.\ In the present article, the generator can explicitly 
be computed as
\[ Af=\inf_{\theta\in\Theta}\Big(G\big(D^2f+\alpha(Df-\theta)(Df-\theta)^T\big)+(Df-\theta)^T\mu\Big)
    -\inf_{\theta\in\Theta}\Big(G\big(\alpha\theta\theta^T\big)-\theta^T\mu\Big), \]
for sufficiently smooth functions $f$,
where $\Theta\subset\Rd$ is a set of available trading strategies and the function
\[ G\colon\R^{d\times d}\to\R,\; a\mapsto\frac{1}{2}\cE\big[\zeta^T a\zeta\big] \]
describes the covariance of the random factors.\ In particular, the asymptotic risk-based 
prices only depend on the covariance of the discrete model but not on the particular choice 
of its distribution.\ In Subsection~\ref{sec:prel.chernoff}, we explain the semigroup approach 
in more detail and fix the precise terminology used throughout the rest of this article.\ 
Furthermore, at the beginning of Section~\ref{sec:proof}, we recall the precise statements 
from~\cite{BDKN22} on which the proofs of the main results in this article are based. 

In the following, the main results of this article are described in more detail.\ In order to 
guarantee the well-posedness of the one-step pricing operators and to exlcude
doubling strategies as the number of intermediate trading periods tends to infinity, 
we first impose a volume constraint on the set of available trading strategies, see Theorem~\ref{thm:main}. 
Since the asymptotic risk-based prices are uniquely determined by the covariance 
of the random factors, a non-degeneracy condition guarantees that the volume constraint 
can be arbitrarily large, see Theorem~\ref{thm:unbounded}.\ This way, we can define 
asymptotic risk-based prices involving unbounded sets of trading strategies as limits
of volume constrained prices. Modeling the set $\Theta$ allows to impose constraints on 
the available trading strategies such as volume or short-selling constraints, non-tradable 
assets, etc.\ However, in the absence of trading constraints, the generator does not depended 
on the gradient and is given by
\[ Af=\inf_{\theta\in\Rd}\Big(G\big(D^2f+\alpha\theta\theta^T\big)-\theta^T\mu\Big) 
    -\inf_{\theta\in\Rd}\Big(G\big(\alpha\theta\theta^T\big)-\theta^T\mu\Big). \]
In particular, the risk-based prices taking into account the attitude of the agent towards 
risk are always dominated by the worst-case prices associated to the $G$-expectation, 
see Corollary~\ref{cor:Rd}.\ So far, in the absence of trading constraints, we did not 
define the asymptotic risk-based prices as the limit of multi-step prices, but as the limit
of volume constrained asymptotic prices. However, under 
additional conditions on the distribution, the asymptotic prices can be obtained as the 
limit of the unconstrained multi-step prices, see Theorem~\ref{thm:unbounded2}.\ Finally, we are interested 
in the asymptotic behaviour of the prices as the risk aversion of the agent tends to infinity.
For a dominated continuous-time framework, it has been shown in~\cite{Rouge2000Pricing,Mania2005Dynamic} 
that the value of the utility maximization problem converges 
to the superhedging price as the risk aversion tends to infinity.\ In a non-dominated 
discrete-time setting, the same result has later been obtained in~\cite{Bartl2019Exponential,Deng2020Utility}, see also~\cite{Carassus2007Optimal,Blanchard2021Convergence} 
for similar results.
Moreover, for finite state spaces and for downside-sensitive preferences, it has been 
shown in~\cite{Cerny2002good} that the risk-based price bounds approach the no-arbitrage 
ones as the set of desirable claims gets smaller.\
In this article, we impose a uniform ellipticity condition on the covariance of the random factors
to show that the asymptotic risk-based prices converge to the worst-case prices associated to the 
$G$-expectation as the risk aversion tends to infinity, see Theorem~\ref{thm:alpha}.

Apart from the novel theoretical insights regarding the converge of the multi-step prices
to the asymptotic risk-based prices, Chernoff-type approximations also provide a tool for 
numerical approximations which is illustrated in Section~\ref{sec:numerics}.\ For instance, 
since the asymptotic risk-based prices only depend on the covariance structure of the model
but not on the particular choice of the distribution, it is sufficient to recursively compute
the indifference prices in the binomial model for every trading period. This also works well 
in the presence of model uncertainty, see Subsection~\ref{sec:numerics.bin} and Subsection~\ref{sec:num.unc.bin}.\ 
Furthermore, the numerical simulations in Subsection~\ref{sec:num.linear} confirm that the
asymptotic prices only depends on the covariance of the increments although the 
precise bounds for the difference between the multi-step prices and the asymptotic prices might 
depend on the choice of the model.\ Convergence rates for Chernoff-type approximations have been 
investigated in~\cite{BJKL23} but addressing this question in the present context is beyond the 
scope of this article.\ So far, we focused on ask prices derived from the view point of the seller.\ 
Similarly, from the perspective of the buyer, one can derive the corresponding bid prices 
$b_t^\infty(f)$ which are related to the ask prices via the equation $b_t^\infty(f)=-a_t^\infty(-f)$.\
In particular, the bid-ask spread for the risk-based prices is always smaller than the the bid-ask
spread for the worst-case prices associated to the $G$-expectation, see Subsection~\ref{sec:bid-ask}.

The rest of this article is organized as follows:\ in Section~\ref{sec:setup}, we introduce the 
market model, the asset distribution given by a sublinear expectation, dynamically consistent 
pricing operators and the necessary terminology regarding strongly continuous convex monotone 
semigroups. Section~\ref{sec:indifference} first introduces the agent's preferences and 
indifference pricing relations before stating the main results.\ Section~\ref{sec:numerics} 
contains several examples in order to illustrate the abstract results including numerical simulations.\
The proofs of the main results are given in Section~\ref{sec:proof}.\ Finally, Appendix~\ref{app:convex}
contains a basic convexity estimates and elementary properties of sublinear expectations 
while Appendix~\ref{app:mgf} contains some exponential moment estimates.

\section{Market model and pricing operators} 
\label{sec:setup}

We consider $d\in\N$ financial assets $X=(X^1,\dots,X^d)^T$ which are traded at discounted 
prices $X_{kh}$ at discrete time-points $\mathcal{T}(h):=\{kh\colon k\in\N_0\}$ for some 
trading period $h>0$.\footnote{Here, the superscript $T$ denotes the transpose of a vector or a matrix.}
All prices/payoffs are discounted by expressing them in terms of a num\'eraire $S^0$ which 
is strictly positive in all possible states at all considered trading times, i.e., the 
discounted value of a random payoff $Z$ at time $t\in\mathcal{T}(h)$ is given by $Z/S^0_t$.
Furthermore, the asset prices start at $X_0^x:=x\in\R^d$ and follow the dynamics 
\begin{equation}\label{ass:dynamics}
 X_{(k+1)h}^x:=X_{kh}^x+h\mu+\sqrt{h}\zeta_{k+1}\quad\mbox{for all }k\in\N,
\end{equation}
where $\mu\in\R^d$ is a deterministic drift and $(\zeta_k)_{k\in\N}$ are i.i.d. random vectors 
$\zeta_k\colon \Omega\to\R^d$ on a sublinear expectation space $(\Omega,\cH,\cE)$ which will
be specified in the following subsection.

\subsection{Asset distribution and sublinear expectations}
\label{sec:expectation}

A sublinear expectation space $(\Omega,\cH,\cE)$ consists of a set $\Omega$, a pointwise ordered 
linear space $\cH$ of random variables $Y\colon\Omega\to\R$ with $c\one_\Omega\in\cH$ and $|Y|\in\cH$ 
for all $c\in\R$ and $Y\in\cH$ and a sublinear expectation $\mathcal{E}\colon\cH\to\R$ satisfying
 \begin{enumerate}
 \item $\cE[c\one_\Omega]=c$ for all $c\in\R$, 
 \item $\cE[Y]\leq\cE[Z]$ for all $Y,Z\in\cH$ with $Y\leq Z$, 
 \item $\cE[Y+Z]\leq \cE[Y]+\cE[Z]$ for all 
  $Y,Z\in\cH$,
  \item $\cE[\lambda Y]=\lambda\cE[Y]$ for all $Y\in\cH$ and $\lambda\ge 0$.
 \end{enumerate}
The sublinear expectation is called continuous from above if $\cE[X_n]\downarrow 0$ for all 
sequences $(X_n)_{n\in\N}$ with $X_n\downarrow 0$. Note that the properties~(i) an~(iii) 
imply cash invariance, i.e., it holds $\cE[Y+c]=\cE[Y]+c$ for all $Y\in\cH$ and $c\in\R$. 
Sublinear expectations were introduced by Peng to incorporate model uncertainty of the asset 
distribution, see~\cite{Peng19} for a detailed discussion. Indeed, the formula 
\begin{equation} \label{eq:dual}
 \cE[Y]:=\sup_{\Q\in\cQ}\E_\Q[Y] \quad\text{for all } Y\in\cH 
\end{equation}
defines a sublinear expectation, where the supremum is taken over an uncertainty set~$\cQ$ 
of probability measures on $(\Omega,\sigma(\cH))$.\ On the other hand, every sublinear expectation 
which is continuous from above admits such a representation, see~\cite[Theorem~1.2.2]{Peng19}. 
Sublinear expectations are also closely related to several other concepts such as coherent risk 
measures in mathematical finance~\cite{ADEH}, upper expectations in robust statistics~\cite{Huber} 
and upper coherent previsions in the theory of imprecise probabilities~\cite{Walley}. 
In a dynamic setting sublinear expectations are linked to BSDEs \cite{elkaroui1997BSDEs,gianin2006gexpectations,soner2012wellposedness}.
Instead of specifying the space $\cH$, we rather assume that it is rich enough to guarantee that
$(\zeta_k)_{k\in\N}\subset\cH^d$ and that all the terms appearing in the following are again elements of $\cH$.
In particular, we assume that $f(\zeta_1,\ldots,\zeta_n)\in\cH$ for all $n\in\N$ and $f\in\Lipb((\R^d)^n)$, 
where $\Lipb((\R^d)^n)$ denotes the space of all bounded Lipschitz-continuous functions $f\colon (\R^d)^n\to\R$.
The random vectors $(\zeta_k)_{k\in\N}$ are supposed to be independent and identically distributed (i.i.d.) 
meaning that 
\[ \cE[f(\zeta_m)]=\cE[f(\zeta_n)] \quad\mbox{for all } f\in\Lipb(\R^d) \mbox{ and } m,n\in\N \] 
and that $\zeta_{n+1}$ is independent of $\zeta_1,\dots,\zeta_n$ for all $n\in\N$, i.e.,  
\[ \cE\left[f(\zeta_1,\dots,\zeta_n,\zeta_{n+1})\right]
    =\cE\left[\cE[f(z,\zeta_{n+1})]\big|_{z=(\zeta_1,\dots,\zeta_n)}\right]
   	\quad\mbox{for all } f\in\Lipb((\R^d)^n\times\R^d). \] 
Furthermore, the random vectors have no mean uncertainty, i.e., 
\begin{equation}\label{eq:nomeanunceratinty}
\cE[a^T\zeta_1]=0\quad\mbox{for all }a\in\R^d.
\end{equation} 
For the sake of illustration, we provide several examples for the uncertain asset distribution
given by the functional $\Cb\to\R,\; f\mapsto\cE[f(\zeta_1)]$, where the space $\Cb:=\Cb(\Rd)$
consists of all bounded continuous functions $f\colon\Rd\to\R$.\ Since we have already assumed that
$f(\zeta_1)\in\cH$ for all $f\in\Lipb:=\Lipb(\Rd)$, the term $\cE[f(\zeta_1)]$ is also well-defined 
for all $f\in\Cb$ if the sublinear expectation $\cE$ is continuous from above. The latter is valid
if the supremum in equation~\eqref{eq:dual} is taken over a tight set of probability measures.

\begin{example} \label{ex:sub.expect} 
\begin{enumerate}
\item Without uncertainty, the measure $\nu:=\P\circ\zeta_1^{-1}$ determines 
 the linear expectation
 \[ \cE[f(\zeta_1)]:=\E_\P[f(\zeta_1)]=\int_{\Rd}f\,\d\nu \quad \text{for all } f\in\Cb \]
 and condition~\eqref{eq:nomeanunceratinty} reduces to $\int_{\R^d} x \,\nu(\d x)=0$.
\item \label{ex:parametric} 
 Let $\nu$ be a probability measure on $(\Rd,\cB(\Rd))$ and $\Lambda\subset\Rd$ be a bounded 
 set.\footnote{As usual, $\cB(\Rd)$ denotes the Borel $\sigma$-algebra on $\Rd$.}
 Perturbing the values of $\zeta_1$ by $\pm\lambda$ leads to a sublinear distribution 
 incorporating parametric uncertainty given by
 \[ \cE[f(\zeta_1)]:=\sup_{\lambda\in\Lambda}\frac{1}{2}\int_{\Rd}f(x+\lambda)+ f(x-\lambda)\,\nu(\d x) 
 \quad \text{for all } f\in \Cb. \]
 Condition~\eqref{eq:nomeanunceratinty} is satisfied if $\int_{\R^d}x\,\nu(\d x)=0$ and for
 $\Lambda=\{0\}$ we recover the linear case without uncertainty. As a particular one-dimensional 
 example of this parametrization, we choose the Dirac measure $\nu:=\delta_0$ and $\Lambda:=[\sigma_0-u,\sigma_0+u]$,
 where $\sigma_0>0$ is a reference volatility and $u\in [0,\sigma_0]$ is the level of uncertainty. 
 Then, the sublinear expectation
 \[ \cE[f(\zeta_1)]:=\sup_{\sigma\in [\sigma_0-u,\sigma_0+u]}\frac{1}{2}\big(f(\sigma)+f(-\sigma)\big)
    \quad \text{for all } f\in\Cb \]
 describes a binomial model with uncertain volatility. We also want to mention that, for every 
 $\lambda\in\Lambda$, the distribution $\nu$ is transformed to the distribution 
 \[ \nu\kappa_\lambda(B):=\int_{\R^d} \kappa_\lambda(x,B)\,\nu(\d x) \quad\text{for all } B\in\cB(\Rd) \]
 through the weak transport plan $\kappa_\lambda(x,B):=\frac{1}{2}\delta_{x+\lambda}(B)+\frac{1}{2}\delta_{x-\lambda}(B)$.
 For uncertainty sets based on weak optimal transport, we refer to \cite{kupper2023risk} and the references therein.
\item \label{ex:nonparametric} 
 In recent years, non-parametric uncertainty has been becoming increasingly popular and has
 been explored extensively, for instance, in the field of distributionally robust optimization, 
 see, e.g,~\cite{bartl2020computational,Blanchet2019Quantifying,Gao2023DRSO,esfahani2018data,
 pflug2007ambiguity,Wozabal2012frameworkambiguity,zhao2018}. The analytical tractability of 
 transport distances such as the Wasserstein distance allows for dual representations 
 and explicit sensitivity analysis, see~\cite{bartl2021sensitivity,Bartl2023SensitivityAW}. 
 Let $\nu$ be a reference distribution with $\int_{\R^d}x\,\nu(\d x)=0$ and 
 $\int_{\R^d}|x|^p\,\nu(\d x)<\infty$ for some $p>2$. We define
 \[ \cE[f(\zeta_1)]:=\sup_{\cW_0(\nu,\tilde\nu)\leq u}\int_{\Rd} f\, \d\tilde\nu 
    \quad\text{for all } f\in\Cb, \]
 where $u\geq 0$ is the level of uncertainty and the transport distance $\cW_0$ is given by 
 \[ \cW_0(\nu,\tilde\nu):=\left(\inf_{\pi\in\Pi_0(\nu,\tilde\nu)}
    \int_{\R^d\times\R^d}|x-y|^p\,\pi(\d x,\d y)\right)^\frac{1}{p} \]
 with the infimum being taken over the set $\Pi_0(\nu,\tilde\nu)$ of all couplings between 
 $\nu$ and~$\tilde{\nu}$ satisfying $\int_{\R^d\times\R^d} x^T  a(y-x)\,\pi(\d x,\d y)=0$ 
 for all symmetric $d\times d$-matrices $a$. Alternatively, the set $\Pi_0(\nu,\tilde\nu)$ 
 can be replaced by the set of all martingale couplings. For details, we refer to~\cite[Section 4.2]{BK22}. 
\end{enumerate}
\end{example}

\subsection{Continuous time limits and Chernoff-type approximations}
\label{sec:prel.chernoff}

So far, we considered asset prices following the dynamics given by equation~\eqref{ass:dynamics}
in a discrete-time framework with fixed step-size $h>0$. Now, we are interested in the limit
behaviour of the asset dynamics as the number of intermediate trading periods tends to infinity. 
Let $t\geq 0$, $x\in\Rd$ and $h_n:=t/n$ for all $n\in\N$. We define $X_0^{n,x}:=x$ and
\begin{equation}
 X_{(k+1)h_n}^{n,x}:=X_{kh}^{n,x}+h_n\mu+\sqrt{h_n}\zeta_{k+1} \quad\mbox{for all } k,n\in\N.
\end{equation}
It follows from Peng's central limit theorem for sublinear expectations, 
see~\cite{HJLP2023,Peng07,Peng08,Peng08b}, that 
\begin{equation} \label{eq:clt}
 \cE\big[f\big(X_t^{n,x}\big)\big]
 =\cE\left[f\Big(x+t\mu+\sqrt{\frac{t}{n}}\zeta_{1}+\cdots+\sqrt{\frac{t}{n}}\zeta_{n}\Big)\right]
 \to\cE[f(B_t^x)] 
\end{equation}
for all $f\in\Cb$, where $B_t^x$ is $G$-normally distributed with $G(a,b):=\frac{t}{2}\cE[\zeta_1^T a\zeta_1]+(x+tb^T)\mu$
for all $a\in\R^{d\times d}$ and $b\in\Rd$. In the linear case  $\cE[\cdot]=\E_\P[\cdot]$,
we simply obtain $B_t^x=x+t\mu+\sqrt{t}\xi$, where $\xi\sim\cN(0,\Sigma)$ is normally distributed 
with covariance matrix $\Sigma:=\E_\P[\zeta_1\zeta_1^T]$. While Peng's definition of the $G$-normal 
distribution relies on the existence and uniqueness of viscosity solutions for the fully nonlinear PDE 
\begin{equation} \label{eq:Gexp}
 \partial_t u(t,x)=G(D^2 u(t,x),Du(t,x)), \quad u(0,x)=f(x)
\end{equation}
by setting $\cE[f(B_t^x)]:=u(t,x)$, in this article, we will take an equivalent semigroup perspective. 
In the linear case, the family $(B_t^x)_{t\geq 0}$ is a Brownian motion whose linear transition 
semigroup is given by the heat semigroup 
\[ (S(t)f)(x):=\E_\P[f(B_t^x)]=\lim_{n\to\infty}\E_\P[f(X^{n,x}_t)] \]
for all $t\geq 0$, $f\in\Cb$ and $x\in\Rd$. The semigroup is uniquely determined by its generator
\[ (Af)(x)=\frac{1}{2}\Tr(\Sigma D^2f(x))+Df(x)^T\mu \]
for all $f\in\Cb^2$ and $x\in\Rd$, where the space $\Cb^2$ consists of all bounded twice continuously
differentiable functions $f\colon\Rd\to\R$ with bounded first and second derivative.\ In the sublinear 
case, the operators
\[ (S(t)f)(x):=\lim_{n\to\infty}\cE[f(X^{n,x}_t)] \]
form a semigroup of sublinear operators $S(t)\colon\Cb\to\Cb$ which is uniquely determined
by its generator 
\[ (Af)(x)=G(D^2f(x),Df(x))=\frac{1}{2}\cE[\zeta_1^T D^2f(x)\zeta_1]+Df(x)^T\mu \]
for all $f\in\Cb^2$ and $x\in\Rd$, see~\cite[Theorem~4.1]{BK22}. Furthermore, 
the unique viscosity solution of equation~\eqref{eq:Gexp} is given by $u(t,x):=(S(t)f)(x)$,
see~\cite[Theorem 6.2]{GNR22}. Subsequently, we explain the semigroup approach in more detail.

Throughout this article, the space $\Cb$ is endowed with the mixed topology between the 
supremum norm $\|\cdot\|_\infty$ and the topology of uniform convergence on compact sets,
i.e., the strongest locally convex topology on $\Cb$ which coincides on $\|\cdot\|_\infty$-bounded
sets with the topology of uniform convergence on compact sets. In particular, for every sequence
$(f_n)_{n\in\N}\subset\Cb$ and $f\in\Cb$, it holds $f_n\to f$ if and only if 
\[ \sup_{n\in\N}\|f_n\|_\infty<\infty \quad\mbox{and}\quad \lim_{n\to\infty}\|f-f_n\|_{\infty,K}=0 \]
for all compact subsets $K\subset\Rd$ and $\|f\|_{\infty,K}:=\sup_{x\in K}|f(x)|$, see~\cite[Proposition~B.2]{GNR22}.
In the following, if not stated otherwise, all limits in $\Cb$ are taken w.r.t. the mixed topology 
and compact subsets are denoted by $K\Subset\Rd$. Although the mixed topology is not metrizable,
it has been observed in~\cite{Nendel22} that, for monotone operators $S\colon\Cb\to\Cb$,
sequential continuity is equivalent to continuity which is further equivalent to continuity 
on norm-bounded sets. For more details on the mixed topology, we refer to~\cite[Appendix~B]{GNR22} 
and the references therein. Since functions are ordered pointwise here, 
 an operator $S\colon\Cb\to\Cb$ is called monotone if $(Sf)(x)\leq (Sg)(x)$ for all $x\in\Rd$ 
and $f,g\in\Cb$ with $f(y)\leq g(y)$ for all $y\in\Rd$ and convex if 
$(S(\lambda f+(1-\lambda)g))(x)\leq \lambda (Sf)(x)+(1-\lambda)(Sg)(x)$ for all $f,g\in\Cb$, 
$\lambda\in [0,1]$ and $x\in\Rd$. The following definition characterizes the semigroups which will
be studied in this article.

\begin{definition}
 A family $(S(t))_{t\geq 0}$ of operators $S(t)\colon\Cb\to\Cb$ is called strongly continuous
 convex monotone semigroup on $\Cb$ if the following conditions are satisfied:
 \begin{enumerate}
  \item $S(t)$ is convex and monotone with $S(t)f_n\downarrow 0$ for all $t\geq 0$ and $f_n\downarrow 0$,
  \item $S(0)f=f$ and $S(s+t)f=S(s)S(t)f$ for all $s,t\geq 0$ and $f\in\Cb$,
  \item $\sup_{t\in [0,T]}\|S(t)r\|_\infty<\infty$ for all $r,T\geq 0$,
  \item $f=\lim_{t\downarrow 0}S(t)f$ for all $f\in\Cb$. 
 \end{enumerate}
 Furthermore, the generator of the semigroup is defined by
 \[ A\colon D(A)\to\Cb,\; f\mapsto\lim_{h\downarrow 0}\frac{S(h)f-f}{h}, \]
 where the domain consists of all $f\in\Cb$ such that the previous limit exists. 
\end{definition}

It has recently been shown by Blessing et al., see~\cite{BDKN22}, that strongly continuous convex 
monotone semigroups are uniquely determined by their so-called upper $\Gamma$-generators defined 
on their upper Lipschitz sets. While this result is convincing due to its generality, in many
applications, the generator $Af$ can only be determined for sufficiently smooth functions $f$.
However, under additional conditions, the semigroup is already uniquely determined by the evaluation
of its generator at smooth functions or even only smooth functions with compact support,
see~\cite{BDKN22,BK22,BKN23}. A precise statement which is sufficient for the applications 
presented in this article is given in Theorem~\ref{thm:unique}. The second main result about strongly 
continuous convex monotone semigroups is that they allow for Chernoff-type approximations of the form 
\begin{equation} \label{eq:chernoff}
 S(t)f=\lim_{n\to\infty}\big(I\big(\tfrac{t}{n}\big)^n f\big)(x).
\end{equation}
Here, the starting point is a family $(I(t))_{t\geq 0}$ of one-step operators $I(t)\colon\Cb\to\Cb$ 
from which we derive the iterated operators $I(t/n)^n:=I(t/n)\circ\ldots\circ I(t/n)$. 
Under suitable stability conditions, the limit in equation~\eqref{eq:chernoff} exists and defines 
a strongly continuous semigroup $(S(t))_{t\geq 0}$ on $\Cb$ which is uniquely determined by the 
infinitesimal behaviour of $(I(t))_{t\geq 0}$. To be precise, it holds 
\[ Af=I'(0)f:=\lim_{h\downarrow 0}\frac{I(h)f-f}{h} \]
for smooth functions $f$ and the previously mentioned comparison principle can be applied. 
Chernoff-type approximations have been studied in~\cite{BK22,BK23,BDKN22,BKN23} and in 
Section~\ref{sec:proof} we recall the precise statement on which the proofs of the main 
results of this article are based.

\begin{example}
 Let $(I(t)f)(x):=\cE[f(x+t\mu+\sqrt{t}\zeta_1)]$ for all $t\geq 0$, $f\in\Cb$ and $x\in\Rd$,
 where the constant drift $\mu\in\Rd$ and the random factors $(\zeta_k)_{k\in\N}\subset\cH$
 are the same as before. It follows from Taylor's formula that
 \[ (I'(0)f)(x)=\frac{1}{2}\cE[\zeta_1^T D^2f(x)\zeta_1]+Df(x)^T\mu \]
 for all $f\in\Cb^2$ and $x\in\Rd$. Furthermore, the stability conditions required for the 
 Chernoff-type approximations are satisfied. Hence, for every $t\geq 0$ and $f\in\Cb$, the limit 
 \[ S(t)f:=\lim_{n\to\infty}I\big(\tfrac{t}{n}\big)^n f \]
 exists and defines a strongly continuous convex monotone semigroup on $\Cb$ which is uniquely 
 determined by its generator
 \[ (Af)(x)=\frac{1}{2}\cE[\zeta_1^T D^2f(x)\zeta_1]+Df(x)^T\mu \]
 for all $f\in\Cb^2$ and $x\in\Rd$. For details, we refer to~\cite[Theorem~4.1]{BK22}\footnote{The
 result in~\cite{BK22} is only stated with $\mu=0$ but the argumentation remains valid when adding 
 a constant drift.}. Moreover, by using that the random factors $(\zeta_k)_{k\in\N}$ are i.i.d., 
 one can show that 
 \begin{equation} \label{eq:chernoff2}
  \big(I\big(\tfrac{t}{n}\big)f\big)(x)
  =\cE\left[f\Big(x+t\mu+\sqrt{\frac{t}{n}}\zeta_{1}+\cdots+\sqrt{\frac{t}{n}}\zeta_{n}\Big)\right]
 \end{equation}
 for all $t\geq 0$, $f\in\Cb$, $x\in\Rd$ and $n\in\N$. Since~\cite[Theorem~6.2]{GNR22} guarantees
 that the unique viscosity solution of equation~\eqref{eq:Gexp} is given by $u(t,x):=(S(t)f)(x)$,
 a random variable $Y\in\cH$ is G-normally distributed with $G(a,b):=\frac{1}{2}\cE[\zeta_1^T a\zeta_1]+b^T\mu$ 
 if and only if $\cE[f(Y)]=(S(1)f)(0)$ for all $f\in\Cb$. In this way, we recover a variant of Peng's 
 central limit theorem as a particular case of a Chernoff-type approximation. In addition, the 
 semigroup approach used in~\cite{BK22} allows to replace the sublinear expectation by a convex 
 expectation without significantly changing the proof. In contrast, the earlier results 
 in~\cite{HJLP2023,Peng07,Peng08,Peng08b} are only stated for the sublinear case. The same 
 is true for the convergence rates in~\cite{BJKL23} based on the semigroup approach in comparison 
 to the ones based on monotone schemes for viscosity solutions in~\cite{HJL21,HL19,Krylov2020,Song20}.

 We conclude this brief illustration of the semigroup approach by picking up the two sublinear expectations
 from Example~\ref{ex:sub.expect}. In the sequel, we denote by $\cE_1$ the sublinear expectation from 
 Example~\ref{ex:sub.expect}(ii) describing parametric uncertainty with $\Lambda:=\{\lambda\in\Rd\colon |\lambda|\leq u\}$
 and by $\cE_2$ the sublinear expectation from Example~\ref{ex:sub.expect}(iii) describing non-parametric 
 uncertainty with the same parameter $u\geq 0$. We define  
 \[ (I(t)f)(x):=\cE_1[f(x+t\mu+\sqrt{t}\zeta_1)] \quad\mbox{and}\quad 
    (J(t)f)(x):=\cE_2[f(x+t\mu+\sqrt{t}\zeta_1)] \] 
 and denote by $(S(t))_{t\geq 0}$ and $(T(t))_{t\geq 0}$ the corresponding semigroups with
 generators $A$ and $B$, respectively. An explicit computation shows that
 \[ (Af)(x)=(Bf)(x)
    =\frac{1}{2}\Tr(\Sigma D^2 f(x))+Df(x)^T\mu+\sup_{|\lambda|\leq u}\frac{1}{2}\Tr(\lambda\lambda^T D^2 f(x)) \]
 with $\Sigma:=\int_{\Rd}yy^T\,\nu(\d y)$ for all $f\in\Cb^2$ and $x\in\Rd$. The linear part on the 
 right-hand side is the generator of the reference model, i.e., a Bachelier model with covariance 
 matrix~$\Sigma$ and drift~$\mu$, whereas the supremum incorporates the model uncertainty. In particular,
 the generator does not depend on the specific type of uncertainty as long as the amount of uncertainty 
 is the same. The comparison principle guarantees that this observation is also valid for the semigroups, 
 i.e., it holds $S(t)f=T(t)f$ for all $t\geq 0$ and $f\in\Cb$. For details, we refer to~\cite[Subsection~4.2]{BK22}. 
\end{example}

\subsection{Dynamically consistent pricing operators}
\label{sec:dyn.consist}

In this subsection, we identify desirable properties for the pricing operators and focus on the pricing 
of European options with payoffs $f(X_t)$.\ As before, we consider discrete trading times 
$\cT(h)=\{kh\colon k\in\N_0\}$ for some trading period $h>0$ and recall that all prices 
are discounted. For every $s,t\in\cT(h)$, we denote by $(p_{s,t}f)(X_s)$ the price at time $s$ 
of the contingent claim $f(X_t)$  in the state $X_s$. For every $s,t,u,v\in\cT(h)$ with 
$s\leq u\leq t$ and $f,g\in\Cb$, we assume that 
\begin{itemize}
 \item[(p1)] $p_{s,t}\colon\Cb\to\Cb$ with $p_{s,s}=\id_{\Cb}$,
 \item[(p2)] $p_{s,t} 0=0$ and $p_{s,t}(f+c\one_{\Rd})=p_{s,t}f+c$ for all $c\in\R$,
 \item[(p3)] $f\leq \tilde f$ implies $p_{s,t}f\leq p_{s,t}\tilde f$, 
 \item[(p4)] $p_{s,u}(p_{u,t}f)=p_{s,t}f$,
 \item[(p5)] $p_{s,s+v}f=p_{t,t+v}f$.
\end{itemize}
The conditions~{\rm (p1)}-{\rm (p3)} have a clear interpretation and are desirable for any pricing 
operator. Since the underlying dynamics $(X_t)_{t\in\cT(h)}$ is a homogeneous Markov process, we 
require in condition~{\rm (p5)} that the pricing operators are also homogeneous, i.e., conditioned 
that the market is at time $s$ and $t$ in the same state, the prices of $f(X_{s+u})$ and $f(X_{t+u})$ 
coincide. In the sequel, the operator $p_{s,s+t}$ is therefore simply denoted by $p_t$.
Furthermore, as we will discuss below, condition~{\rm (p4)} guarantees the consistency of 
the extended pricing operators which reduces to the semigroup property, i.e., $p_s(p_t f)=p_{s+t}f$ 
for all $s,t\in\cT(h)$ and $f\in\Cb$. In particular, the pricing operators $(p_t)_{t\in\cT(h)}$ 
are fully determined by the one-step pricing operators $I(h)f:=p_h f$ and by the equation
\[ p_{kh}f=I(h)^{k}f \quad\mbox{for all } k\in\N. \]

We next discuss an extension of pricing operators to path-dependent options. To do so, let $\mathcal{X}:=\bigcup_{t\in\cT(h)}\mathcal{X}_t$, where $\cX_{kh}$ denotes the space of all bounded continuous
functions $g\colon (\Rd)^{k+1}\to\R$. Then, $g\in \cX$ represents a path-dependent option $g(X_0,\ldots,X_{kh})$ for some $k\in\N_0$.
Suppose that, for every $s\in\cT(h)$, there exists $\hat{p}_{s}\colon\cX\to\cX_s$ with
\[ \big(\hat{p}_{s}g\big)(x_0,\dots, x_s)=\big(p_{t-s}f\big)(x_s) \]
for all $g\in\cX$ of the form $g(x_0,\ldots, x_t,\dots, x_u)=f(x_t)$ 
for $t,u\in\cT(h)$ with $s\le t\le u$ and $f\in\Cb$. In addition, for every $s,t\in\cT(h)$ with $s\leq t$ and $g,\tilde g\in\cX$, we assume that 
\begin{itemize}
 \item[($\hat{\rm p}$1)] $\hat{p}_{s}0=0$ and $\tilde g\in\cX_s$ implies $\hat{p}_{s}(g+\tilde g)=\hat{p}_{s}g+\tilde g$,
 \item[($\hat{\rm p}$2)] $g\leq \tilde g$ implies $\hat{p}_{s}g\leq\hat{p}_{s}\tilde g$,
\item[($\hat{\rm p}$3)] $\hat{p}_{t}g\leq\hat{p}_{t} \tilde g$ implies $\hat{p}_{s}g\leq\hat{p}_{s}\tilde g$. 
\end{itemize}
 If condition~($\hat{\rm p}$3) were violated, there would exist a time and a state where $\tilde g$ is priced strictly higher than $g$, even though $\tilde g$ has a lower price than $g$ in all possible states at some future time. This would result in time inconsistencies in the prices. Moreover, it is well known that condition~($\hat{\rm p}$3) implies the dynamic programming principle or Bellman's principle, see e.g.~\cite{CK11} and the respective references after Definition 2.2. In our context, the following statement holds. 
\begin{lemma}
Under the assumption that conditions~($\hat{\rm p}${\rm 1}) and ($\hat{\rm p}${\rm 2}) are satisfied, 
condition~($\hat{\rm p}${\rm 3}) is equivalent to 
\begin{itemize}\label{lem:ext pricing}
 \item[{\rm ($\hat{\rm p}$3')}]  $\hat{p}_{s}(\hat{p}_{t}g)=\hat{p}_{s}g$ for all $s,t\in\cT(h)$ 
with $s\leq t$ and $g\in\cX$.
\end{itemize}
In particular,  $p_s(p_t f)=p_{s+t}f$ 
for all $s,t\in\cT(h)$ and $f\in\Cb$.
\end{lemma}    
\begin{proof}
 Suppose that condition~($\hat{\rm p}$3) is satisfied. Let $s,t\in\cT(h)$ with $s\leq t$ 
 and $g\in\cX$. Since condition~($\hat{\rm p}$1) implies $\hat{p}_{t}g=\hat{p}_{t}(\hat{p}_{t}g)$, 
 we obtain from ($\hat{\rm p}$3) that $\hat{p}_{s}(g)=\hat{p}_{s}(\hat{p}_{t} g)$. 
 
Conversely, suppose that condition~($\hat{\rm p}$3') holds. Let $s,t\in\cT(h)$ with $s\leq t$ and $g,\tilde g\in\cX$ with $\hat{p}_{t}g\leq\hat{p}_{t} \tilde g$. Using condition~($\hat{\rm p}$2), we obtain $\hat{p}_{s}g=\hat{p}_{s}(\hat{p}_{t}g)\leq\hat{p}_{s}(\hat{p}_{t}\tilde g)=\hat{p}_{s}\tilde g$.

As for the second part, let $s,t\in\cT(h)$ and $f\in\Cb$.
For $g(x_0,\dots,x_{s+t}):=f(x_{s+t})$, we have $\tilde{f}(x_s):=(p_t f)(x_s)=(\hat p_s g)(x_0,\dots,x_s)$. Hence, it follows from condition~($\hat{\rm p}$3') that
\[
p_s(p_t f)=p_s \tilde f=\hat p_0 \tilde g=\hat p_0(\hat p_s(g)=\hat p_0 g=p_{t+s}f,
\]
where $\tilde g(x_0,\dots,x_s):=\tilde f(x_s)$ and therefore $\tilde g= \hat p_s g$.
\end{proof}

\begin{remark}
Let $(\hat p_s)_{s\in \cT(h)}$ be a family of path-dependent pricing operators $\hat p_s\colon \cX\to\cX_s$ which satisfies ($\hat{\rm p}$1)-($\hat{\rm p}$3). Then, its restriction
$p_t\colon \Cb\to\Cb$ given by 
\[
p_t f:=\hat p_0 g\quad\mbox{where }g(x_0,\dots,x_t):=f(x_t)
\]
of homogeneous pricing operators for options with payoff functions $f(X_t)$,
satisfy 
\begin{itemize}
\item[(p1')] $p_{0}=\id_{\Cb}$,
\item[(p2')] $p_{t}0=0$ and $p_{t}(f+c\one_{\Rd})=p_{t}f+c$ for all $c\in\R$,
\item[(p3')] $f\leq \tilde f$ implies $p_{t}f\leq p_{t}\tilde f$, 
\item[(p4')] $p_s(p_t f)=p_{s+t}f$,
\end{itemize}
for all $s,t\in\cT(h)$. Here, the conditions (p1')-(p3') follow directly from the definition, while (p4') is a consequence of Lemma~\ref{lem:ext pricing}.  
\end{remark}

As a result of the previous discussion, we obtain that a homogeneous pricing operator $(p_t)_{t\in\cT(h)}$, 
which allows for an extension $(\hat p_t)_{t\in\cT(h)}$ of pricing operators for path-dependent options, 
necessarily has to satisfy the semigroup property (p4').
In this sense, the semigroup property is necessary to avoid 
time inconsistencies in the corresponding prices. In particular, the pricing operator is given by the one-step pricing operators $I(h)f:=p_h f$.

We finally remark that an extension to path-dependent pricing operators exists under rather mild conditions. For instance, if
the mapping $(x_0,\dots,x_{kh})\mapsto p_h(g(x_0,\dots,x_{kh},\cdot))$ is continuous for all $k\in\N$ and any bounded continuous function $g\colon (\Rd)^{k+2}\to\R$, see e.g.~\cite[Proposition 5.5]{DKN18}, then for every $g\in \cX_t$ for some $t\in\cT(h)$, it follows that the operators $(\hat p_s)_{s\in\cT(h)}$ given by the backward recursion 
\[
\begin{array}{cl}
\hat p_s g:=g &\mbox{for }s\ge t \\
\hat p_s g:= p_h \big(\hat p_{s+h} g\big) &\mbox{for }s< t,
\end{array}
\]
have the desirable properties.

\section{Agent's preferences and indifference pricing} 
\label{sec:indifference}

We now introduce agent's preferences by considering an agent who measures her risk exposition 
by the entropic risk measure with risk aversion parameter $\alpha\in (0,\infty)$, i.e., the 
agent's risk on the random loss $Y\in\cH$ is given by
\[ \rho[Y]:=\frac{1}{\alpha}\log\big(\cE\big[e^{\alpha Y} \big]\big)\in(-\infty,\infty], \]
where $(\Omega,\cH,\cE)$ is a sublinear expectation space incorporating model uncertainty
of the asset distribution. Here, we consider risk measures as functionals defined on losses 
rather than on positions, i.e., the risk of a position $Z$ is given by $\rho[-Z]$.
In order to develop our indifference pricing framework, we first focus on ask pricing operators 
representing the seller's price of European contingent claims and the corresponding bid prices 
will then be derived in Section~\ref{sec:bid-ask}. Recall that the asset dynamics 
$(X_t^x)_{t\in\cT(h)}$ with trading period $h>0$ have already been specified at the beginning 
of Section~\ref{sec:setup}. 
Hence, the ask price $a_{h,X^x_t}(f)$ for the contingent claim $f(X^x_{t+h})$ given the asset 
price $X^x_t$ at time $t\in\cT(h)$ is determined by the indifference pricing relation
\[ \inf_{\theta\in\Theta}\rho\left[f(X^x_{t+h})-a_{h,X^x_t}(f)-\theta^T\left(X^x_{t+h}-X^x_t\right)\right]
    =\inf_{\theta\in\Theta}\rho\left[-\theta^T\left(X^x_{t+h}-X^x_t\right)\right], \]
where $\Theta\subset\Rd$ contains all available trading strategies. This relation should be read 
in the following way: assuming that the agent can always trade on the market to reduce her risk 
exposure, the quantity $a_{h,X^x_t}(f)$ makes the agent indifferent between selling the derivative 
at this price or keeping it.
Furthermore, the set $\Theta$ of available trading strategies can a priori model any type of 
constraint. For example, we could consider $\Theta:=\Rd$ if the agent can trade without constraints 
all assets in the market or $\Theta:=\R^m$ for some $m<d$ if, for any reason, the agent cannot trade 
some of the assets. The set $\Theta$ could also be bounded if volume constraints are imposed.
Note that, in principle, by modelling $\Theta$ in a suitable way, not all the components of the 
asset process need to be assets on the market so that the derivative could also depend on some 
external factors.

Since the factors $(\zeta_k)_{k\in\N}$ are i.i.d., we obtain that the ask prices $a_{h,x}(f)$ 
for one trading period are fully determined by the equation
\[ \tilde{\rho}_{h,x}[f-a_{h,x}(f)]=\tilde{\rho}_{h,x}[0], \]
where the trading adjusted risk functional is given by
\begin{equation} \label{def:rhotilde}
 \tilde{\rho}_{h,x}[f]:=\inf_{\theta\in\Theta}\rho\left[f(x+h\mu+\sqrt{h}\zeta_1)-\theta^T (h\mu+\sqrt{h}\zeta_1)\right].
\end{equation}
Furthermore, by applying the cash invariance on the deterministic number $a_{h,x}(f)\in\R$, 
it follows that the one-step pricing operator is given by 
\begin{equation}\label{eq:def_I}
(I(h)f)(x):=a_{h,x}(f)=\tilde{\rho}_{h,x}[f]-\tilde{\rho}_{h,x}[0]
\end{equation}
for all $f\in\Cb$ and $x\in\Rd$. 
Under reasonable assumptions specified in Section~\ref{sub:main}, one can show that 
$I(h)\colon\Cb\to\Cb$ and therefore, as discussed in Subsection~\ref{sec:dyn.consist},  
the time consistent multi-step pricing operators are given by
\begin{equation} \label{def:a multi}
 a^k_{kh,x}(f):=(I(h)^k f)(x) \quad\mbox{for all } k\in\N.
\end{equation}
Similar to the worst-case asset dynamics in Subsection~\ref{sec:prel.chernoff}, we are now 
interested in the limit behaviour of the ask prices as the number of intermediate trading 
periods tends to infinity. Let $t\geq 0$ and $h_n:=t/n$ for all $n\in\N$. Then, the limit
\begin{equation} \label{eq:price.limit}
 a_{t,x}^\infty(f):=\lim_{n\to\infty}a^n_{nh_n,x}(f)=\lim_{n\to\infty}\big(I\big(\tfrac{t}{n}\big)^n f\big)(x)
\end{equation}
defines the time-consistent asymptotic risk-based price of a claim with payoff function $f$.

Before stating the main results, we want to explain the relation between local and global 
indifference pricing in the present framework. Our formalization of the indifference pricing 
relation might appear slightly different from classical indifference pricing because one usually
starts from a risk measure that is defined globally on the entire path of the asset process 
and the hedging strategy. However, although the pricing operator here is defined locally by 
a one-step indifference pricing principle, its concatenation in equation~\eqref{def:a multi} 
again satisfies an indifference pricing relation.
Indeed, since the entropic risk measure is time-consistent, we obtain  
\[ I\big(\tfrac{t}{n}\big)^n f=\tilde{I}\big(\tfrac{t}{n}\big)^n f -nc\big(\tfrac{t}{n}\big)
    =\tilde{I}\big(\tfrac{t}{n}\big)^n f-\tilde{I}\big(\tfrac{t}{n}\big)^n 0,\]
where $(\tilde{I}(t)f)(x):=\tilde{\rho}_{t,x}[f]$. Hence, equation~\eqref{def:a multi} and the 
cash invariance of $\tilde{I}(\frac{t}{n})^n$ yield
\begin{equation} \label{indif:global}
 \Big(\tilde{I}\big(\tfrac{t}{n}\big)^n\big(f-a^n_{t,x}(f)\big)\Big)(x)
 =\big(\tilde{I}\big(\tfrac{t}{n}\big)0\big)(x).
\end{equation}
Using that the random factors $(\zeta_k)_{k\in\N}$ are i.i.d., we obtain the global indifference relation
\[ \inf_\theta\rho\left[f(X^x_t)-a^n_{t,x}(f)-(\theta\cdot X^x)_t\right]
    =\inf_\theta\rho\left[-(\theta\cdot X^x)_t\right], \]
where $(\theta\cdot X^x)_t:=\sum_{k=1}^n\theta_k(X^x_{kh_n}-X^x_{(k-1)h_n})$ with $h_n:=t/n$ and 
the infima are taken over all $\Theta$-valued processes $\theta=(\theta_1,\dots,\theta_n)$ 
such that $\theta_k$ is $X^x_{k-1}$-measurable for all $k=1,\dots,n$.

\subsection{Main results}\label{sub:main}

Recall that $\mu\in\Rd$ is a constant drift and that $(\zeta_k)_{k\in\N}\subset\cH^d$
is an i.i.d. sequence of random variables defined on a sublinear expectation space 
$(\Omega,\cH,\cE)$.\ So far, we did not specify the space $\cH$ but assumed it to be
rich enough to guarantee that all the appearing expectations are well defined. 
In order to state and prove the results in this section, we define $\zeta:=\zeta_1$ 
and impose the following conditions.

\begin{assumption} \label{ass:main}
 Let $\Theta\subset\Rd$ be a closed convex set including zero.\ Suppose that~$\cH$ contains 
 all $\zeta$-measurable functions $X\colon\Omega\to\R$ satisfying $|X|\leq ae^{b|\zeta|}$ 
 for some $a,b\geq 0$, where $|\cdot|$ denotes the Euclidean norm. In addition, for every 
 $a\in\Rd$ and $b\geq 0$,
 \[ \cE[a^T\zeta]=0 \quad\mbox{and}\quad 
    \lim_{c\to\infty}\cE\big[e^{b|\zeta|}\one_{\{|\zeta|\geq c\}}\big]=0. \] 
\end{assumption}

If the expectation $\cE[\,\cdot\,]=\E_\P[\,\cdot\,]$ is linear, one can choose $\cH=L^1(\P)$ 
and the previous conditions reduce to $\E_\P[\zeta]=0$ and $\E_\P[e^{b|\zeta|}]<\infty$ 
for all $b\geq 0$. Furthermore, the condition $\cE[a^T\zeta]=0$ states that the mean is 
not uncertain.\ When passing from the multi-step prices to the asymptotic risk-based prices
the number of intermediate trading times tends to infinity. Hence, in order to exclude 
doubling strategies, we impose a volume constraint on the trading sets by considering
$\Theta_R:=\Theta\cap B_R(0)$, where $B_R(x):=\{y\in\Rd\colon |x-y|\leq R\}$ for all $R\geq 0$ 
and $x\in\Rd$. This constraint also guarantees that the one-step pricing operators
\[ (I_R(t)f)(x):=\inf_{\theta\in\Theta_R}\rho[f(x+t\mu+\sqrt{t}\zeta)-\theta^T(t\mu+\sqrt{t}\zeta)]
    -\inf_{\theta\in\Theta_R}\rho[-\theta^T(t\mu+\sqrt{t}\zeta)]. \]
are well defined for all $R,t\geq 0$, $f\in\Cb$ and $x\in\Rd$. The next theorem shows that
the asymptotic risk-based prices are well defined and fully determined by the covariance 
\[ G\colon\R^{d\times d}\to\R,\; a\mapsto\frac{1}{2}\cE[\zeta^T a\zeta] \]
of the random factors and the deterministic drift. We define
\begin{align} 
 G_\theta(a,b) &:=\frac{1}{2}\cE\left[\zeta^T a\zeta+\alpha|(b-\theta)^T\zeta|^2\right]
    +(b-\theta)^T\mu \label{eq:Gtheta} \\
 &=G\left(a+\alpha(b-\theta)(b-\theta)^T\right)+(b-\theta)^T\mu \nonumber
\end{align}
for all $a\in\R^{d\times d}$, $b\in\Rd$ and $\theta\in\Theta$. Moreover, we recall that 
$\Cb^2$ contains all bounded twice continuously differentiable functions $f\colon\Rd\to\R$ 
with bounded first and second derivative and that all limits in $\Cb$ are taken w.r.t. 
the mixed topology.

\begin{theorem} \label{thm:main}
 Suppose that Assumption~\ref{ass:main} is satisfied. Then, for every $R\geq 0$, the limit
 \[ S_R(t)f:=\lim_{n\to\infty}I_R\big(\tfrac{t}{n}\big)^n f \]
 of the volume constrained multi-step ask pricing operators exists for all $t\geq 0$ and $f\in\Cb$. 
 Furthermore, the family $(S_R(t))_{t\geq 0}$ is a strongly continuous convex monotone semigroup
 on $\Cb$ which is uniquely determined by its generator satisfying $\Cb^2\subset D(A_R)$ and 
 \begin{align}
  (A_Rf)(x) &= \inf_{\theta\in\Theta_R}G_\theta(D^2f(x),Df(x))-\inf_{\theta\in\Theta_R}G_\theta(0,0)
    \nonumber \\
  &=\inf_{\theta\in\Theta_R}\Big(\frac{1}{2}
    \cE\big[\zeta^T D^2f(x)\zeta+\alpha |(Df(x) - \theta)^T\zeta|^2\big] + (Df(x) - \theta)^T \mu\Big) 
    \nonumber \\
  &\quad\; -\inf_{\theta\in\Theta_R}\Big(\frac{\alpha}{2}\cE\big[|\theta^T\zeta|^2\big] - \theta^T \mu\Big)
    \label{eq:gen.bounded}
 \end{align}
 for all $f\in\Cb^2$ and $x\in\Rd$.
\end{theorem}

The proof is given in Subsection~\ref{sec:proof1}.\ Without additional conditions, 
the volume constraint is necessary to prevent the two infima in equation~\eqref{eq:gen.bounded} 
from taking the value $-\infty$. However, in case that there exists $\delta>0$ with 
\begin{equation} \label{eq:non-deg} 
 \cE\big[|\theta^T\zeta|^2\big]\geq\delta |\theta|^2 
 \quad\mbox{for all } \theta\in\Theta,
\end{equation}
one can always restrict the infima to a bounded set which might depend on $f$. This 
allows to take the limit $R\to\infty$ in equation~\eqref{eq:gen.bounded} 
and the next theorem shows that this transfers to the semigroups $(S_R(t))_{t\geq 0}$. 
Hence, we can define asymptotic risk-based prices involving unbounded sets of trading 
strategies as limits of volume constraint prices. Moreover, for $\Theta=\Rd$ it is 
sufficient to require that there exists $\delta>0$ with 
\begin{equation} \label{eq:non-deg2} 
 \cE\big[|\theta^T\zeta|^2\big]\geq\delta |\theta|^2 
 \quad\mbox{for all } \theta\in\Rd \mbox{ with } \theta^T\mu\neq 0.
\end{equation}

\begin{theorem} \label{thm:unbounded}
 Suppose that Assumption~\ref{ass:main} and condition~\eqref{eq:non-deg} are valid.
 Then, the limit 
 \[ S(t)f:=\lim_{R\to\infty}S_R(t)f \]
 of the volume constrained prices exists for all $t\geq 0$ and $f\in\Cb$.\ Furthermore, 
 the family $(S(t))_{t\geq 0}$ is a strongly continuous convex monotone semigroup on $\Cb$ 
 which is uniquely determined by its generator satisfying $\Cb^2\subset D(A)$ and 
 \[ (Af)(x)=\inf_{\theta\in\Theta}G_\theta(D^2f(x),Df(x))-\inf_{\theta\in\Theta}G_\theta(0,0)
    \quad\mbox{for all } f\in\Cb^2 \mbox{ and } x\in\Rd. \]
 Moreover, for $\Theta:=\Rd$, it is sufficient to require condition~\eqref{eq:non-deg2}
 instead of condition~\eqref{eq:non-deg}. 
\end{theorem}

The proof is given in Subsection~\ref{sec:proof2}. In the case $\Theta=\Rd$, the asymptotic 
prices do not dependent on the first derivative $Df$ and are dominated by the $G$-expectation 
which has previously been introduced in Subsection~\ref{sec:prel.chernoff}.\ Hence, while 
pricing with a $G$-expectation corresponds to pricing according the worst-case measure in 
the ambiguity set, the risk-based framework leads to a mitigation of the worst-case bounds 
by taking into account the attitude of the agent towards risk.

\begin{corollary} \label{cor:Rd}
 Let $\Theta:=\Rd$ and suppose that Assumption~\ref{ass:main} and condition~\eqref{eq:non-deg2} 
 are satisfied. Then, denoting by $(S(t))_{t\geq 0}$ the semigroup from Theorem~\ref{thm:unbounded}, 
 we obtain
 \begin{align*}
  (Af)(x) &=\inf_{\theta\in\Rd}\bigg(\frac{1}{2}\cE\big[\zeta^T D^2f(x)\zeta+\alpha |\theta^T\zeta|^2\big]-\theta^T\mu\bigg) 
    -\inf_{\theta\in\Rd}\bigg(\frac{\alpha}{2}\cE\big[|\theta^T\zeta|^2\big]-\theta^T \mu\bigg) \\
  &\leq\frac{1}{2}\cE[\zeta^T D^2f(x)\zeta]
 \end{align*}
 for all $f\in\Cb^2$ and $x\in\Rd$. Hence, it holds $S(t)f\leq T(t)f$ for all $t\geq 0$ 
 and $f\in\Cb$, where the strongly continuous convex monotone semigroup $(T(t))_{t\geq 0}$
 on $\Cb$ is given by 
 \[ T(t)f:=\lim_{n\to\infty}J\big(\tfrac{t}{n}\big)f \quad\mbox{with}\quad
    (J(t)f)(x):=\cE[f(x+\sqrt{t}\zeta)]. \]
\end{corollary}
\begin{proof}
 Let $f\in\Cb^2$ and $x\in\Rd$. Since $\Theta=\Rd$, we can substitute $\theta$ by $Df(x)+\theta$ to obtain
 \[ \inf_{\theta\in\Rd}G_\theta\left(D^2f(x),Df(x)\right)=\inf_{\theta\in\Rd}G_\theta\big(D^2f(x),0\big). \]
 In addition, for every $\theta\in\Rd$, the sublinearity of $\cE[\,\cdot\,]$ implies 
 \[ \frac{1}{2}\cE\big[\zeta^T D^2f(x)\zeta+\alpha |\theta^T\zeta|^2\big]-\theta^T\mu
    \leq\frac{1}{2}\cE\big[\zeta^T D^2f(x)\zeta\big] 
    +\frac{\alpha}{2}\cE\big[|\theta^T\zeta|^2\big]-\theta^T\mu \]
 and therefore $(Af)(x)\leq\frac{1}{2}\cE[\zeta^T D^2f(x)\zeta]$. Since the family $(T(t))_{t\geq 0}$ 
 is a strongly continuous convex monotone semigroup on $\Cb$ with generator
 \[ (Bf)(x)=\frac{1}{2}\cE[\zeta^T D^2f(x)\zeta] \quad\mbox{for all } f\in\Cb^2 \mbox{ and } x\in\Rd, \]
 it follows from Theorem~\ref{thm:unique} that $S(t)f\leq T(t)f$ for all $t\geq 0$ and $f\in\Cb$.
\end{proof}

In one dimension, condition~\eqref{eq:non-deg} is valid if and only if $\mu=\zeta=0$ 
or $\Theta=0$ or $\cE[|\zeta|^2]>0$. The first case is trivial and the second case corresponds
to the $G$-expectation. Moreover, writing $\cE[\,\cdot\,]=\sup_{\Q\in\cQ}\E_\Q[\,\cdot\,]$, 
the third case occurs if there exists $\Q\in\cQ$ with $\Q(\zeta\neq 0)>0$. 
Hence, in one dimension, condition~\eqref{eq:non-deg} is satisfied in all relevant examples.
Furthermore, in multi dimensions, condition~\eqref{eq:non-deg2} means that the variance
of the increment $\theta^T\zeta$ is non zero for any strategy which also non trivially  
invests in the drift. Since $\zeta$ has mean zero and thus the chance of loosing the 
investment exists, this means that the agent can not use the drift in order to reduce her 
risk infinitely.

So far, we defined the asymptotic risk-based prices corresponding to the case that no
trading constraints are imposed as the limit of asymptotic risk-based prices
corresponding to the case that the trading strategies are restricted to a bounded set. 
The question arises whether, in the absence of trading constraints, the asymptotic 
risk-based prices can also be obtained directly as the limit of unconstrained multi-step indifference 
prices. In order to achieve this approximation, we assume that there exist $M\geq 0$ and 
$t_1>0$ with 
\begin{equation} \label{eq:mgf1}
 \log\big(\cE[e^{t|\zeta|^2}]\big)\leq Mt \quad\mbox{for all } t\in [0,t_1].
\end{equation}
In addition, for every $C\geq 0$, there exist $t_2>0$ and $R\geq 0$ with 
\begin{equation} \label{eq:mgf2}
 \log\big(\cE[e^{\sqrt{t}\theta^T\zeta-t|\zeta|^2}]\big)\geq C|\theta|t
 \quad\mbox{for all } t\in [0,t_2] \mbox{ and } |\theta|\geq R.
\end{equation}
In particular, applying condition~\eqref{eq:mgf2} with $C:=|\mu|$ yields $t_0>0$ such that 
\[ (I(t)f)(x):=\inf_{\theta\in\Theta}\rho[f(x+t\mu+\sqrt{t}\zeta)-\theta^T(t\mu+\sqrt{t}\zeta)]
    -\inf_{\theta\in\Theta}\rho[-\theta^T(t\mu+\sqrt{t}\zeta)] \]
is well-defined for all $t\in [0,t_0]$, $f\in\Cb$ and $x\in\Rd$. Previously, we only
imposed conditions on the first and second moments of $\zeta$ which uniquely determine
the asymptotic risk-based prices.\ This is due to the fact that strongly continuous convex
monotone semigroups are uniquely determined by their generators. In particular,
the convergence in Theorem~\ref{thm:unbounded} is derived from the convergence of
the generators which only depend on the covariance of $\zeta$ but not on further 
information about its distribution. In contrast to the asymptotic prices, the one-step 
prices depend on the particular distribution of $\zeta$ which explains the necessity 
of imposing additional conditions on the exponential moments.

\begin{theorem} \label{thm:unbounded2}
 Let $\Theta:=\Rd$ and let Assumption~\ref{ass:main} and the conditions~\eqref{eq:mgf1} 
 and~\eqref{eq:mgf2} be satisfied. Then, denoting by $(S(t))_{t\geq 0}$ the semigroup from 
 Theorem~\ref{thm:unbounded}, we obtain
 \[ S(t)f=\lim_{n\to\infty}I\big(\tfrac{t}{n}\big)^n f 
    \quad\mbox{for all } t\geq 0 \mbox{ and } f\in\Cb. \]
\end{theorem}

The proof is given in Subsection~\ref{sec:proof3} and it relies on the fact that doubling strategies are automatically excluded by the additional conditions on the risk measure.\ Furthermore, we show in Appendix~\ref{app:mgf} that the conditions~\eqref{eq:mgf1} and~\eqref{eq:mgf2} are satisfied for bounded symmetric distributions and for families of normal distributions. These distributions naturally appear in numerical implementations 
of the iterative scheme, see Section~\ref{sec:numerics}.

So far, the risk aversion of the agent has been described by a fixed parameter $\alpha>0$
which did not appear in the notation. However, the generator and thus the 
corresponding semigroup clearly depend on the choice of this parameter.\ Subsequently, 
we denote by $(S_\alpha(t))_{t\geq 0}$ the semigroup from Theorem~\ref{thm:unbounded} 
previously denoted by $(S(t))_{t\geq 0}$ and by $A_\alpha f$ its generator previously 
denoted by $Af$. Corollary~\ref{cor:Rd} states that, in the absence of trading constraints
and for any $\alpha>0$, the asymptotic risk-based prices are dominated by the worst-case 
prices corresponding to the $G$-expectation.\ We now show that this upper bound is achieved 
as the risk aversion of the agent tends to infinity if there exists $\delta>0$ with
\begin{equation} \label{eq:strict_ell}
 -\cE\big[-|\theta^T\zeta|^2\big]\geq\delta |\theta|^2 
 \quad\mbox{for all } \theta\in\Rd \mbox{ with } \theta^T\mu\neq 0.
\end{equation}
Condition~\eqref{eq:strict_ell} guarantees that condition~\eqref{eq:non-deg} is also 
valid since Lemma~\ref{lem:E}(iv) implies 
\[ \cE[|\theta^T\zeta|^2]\geq -\cE[-|\theta^T\zeta|^2] \quad\mbox{for all } \theta\in\Rd. \]
In particular, Theorem~\ref{thm:unbounded} and Corollary~\ref{cor:Rd} can be applied.
For a linear expectation both conditions are clearly equivalent but the same is not
true in the sublinear case.

\begin{theorem} \label{thm:alpha}
 Let $\Theta:=\Rd$ and suppose that Assumption~\ref{ass:main} and condition~\eqref{eq:strict_ell} 
 are satisfied. Then, as the risk aversion of the agent tends to infinity, the limit 
 \[ S(t)f:=\lim_{\alpha\to\infty}S_\alpha(t)f \] 
 of the unconstrained asymptotic risk-based prices exists for all $t\geq 0$ and $f\in\Cb$.\
 Moreover, the family $(S(t)_{t\geq 0}$ is a strongly continuous convex monotone semigroup 
 on $\Cb$ which is uniquely determined by its generator satisfying $\Cb^2\subset D(A)$ and 
 \[ (Af)(x)=\frac{1}{2}\cE\big[\zeta^T D^2f(x)\zeta\Big] \quad\mbox{for all } f\in\Cb^2 \mbox{ and } x\in\Rd. \]
 In particular, the limit of the risk-based prices coincides with the worst-case prices, i.e., 
 \[ S(t)f=T(t)f \quad\mbox{for all } t\geq 0 \mbox{ and } f\in\Cb, \]
 where the strongly continuous convex monotone semigroup $(T(t))_{t\geq 0}$ on $\Cb$ is given by 
 \[ T(t)f:=\lim_{n\to\infty}J\big(\tfrac{t}{n}\big)f \quad\mbox{with}\quad
    (J(t)f)(x):=\cE[f(x+\sqrt{t}\zeta)]. \]
\end{theorem}

The proof is given in Subsection~\ref{sec:proof4}.\ We conclude this section with
a brief discussion of the difference between the conditions~\eqref{eq:non-deg} 
and~\eqref{eq:strict_ell} using the binomial model and the normal distribution
as illustrative examples.

\begin{remark} \label{rem:alpha}
 First, we consider a one-dimensional uncertain binomial model
 \[ \cE[f(\zeta)]:=\sup_{\sigma\in [\underline{\sigma},\overline{\sigma}]}\frac{1}{2}\big(f(\sigma)+f(-\sigma)\big)
    \quad\mbox{for all } f\in\Cb, \]
 where $0\leq\underline{\sigma}\leq\overline{\sigma}$ are fixed parameters.\
 Condition~\eqref{eq:non-deg} is equivalent to $\overline{\sigma}>0$ meaning that the
 ambiguity set contains at least one non deterministic model and condition~\eqref{eq:strict_ell} 
 is equivalent to $\underline{\sigma}>0$ meaning that all the models are non deterministic
 and their volatility is uniformly bounded from below. 
 
 Second, let $\Lambda$ be a bounded set of positive semi-definite symmetric $d\times d$-matrices and 
 \[ \cE[f(\zeta)]:=\sup_{\Sigma\in\Lambda}\int_{\Rd}f(\sigma y)\,\cN(0,\one)(\d y) \quad\mbox{for all } f \in\Cb, \]
 where $\cN(0,\one)$ denotes the $d$-dimensional standard normal distribution and $\sigma\in\R^{d\times d}$
 is any matrix with $\sigma\sigma^T=\Sigma$. For every $\theta\in\Rd$, 
 \[ \cE[|\theta^T\zeta|^2]
    =\sup_{\Sigma\in\Lambda}\int_{\Rd}\bigg(\sum_{i,j=1}^d\sigma_{i,j}\theta_i y_j\bigg)^2\,\cN(0,\one)(\d y) 
    =\sup_{\Sigma\in\Lambda}|\sigma^T\theta|^2=\sup_{\Sigma\in\Lambda}\theta^T\Sigma\theta. \]
 Hence, condition~\eqref{eq:non-deg} means that, for every $\theta\in\Rd$, there exists $\Sigma\in\Lambda$ 
 with $\theta^T\Sigma\theta\geq\delta |\theta|^2$. However, none of the matrices has to be positive definite, i.e., none of the linear models has to satisfy condition~\eqref{eq:non-deg}.
 On the other hand, condition~\eqref{eq:strict_ell} is satisfied if and only if $\Lambda$
 is a set of uniformly positive definite matrices, i.e., it holds
 $\inf_{|\theta|=1}\inf_{\Sigma\in\Lambda}\theta^T\Sigma\theta>0$. Hence, all
 the linear models have to satisfy condition~\eqref{eq:strict_ell} with a uniform parameter. 
\end{remark}

\section{Examples and numerical illustrations} 
\label{sec:numerics}

In this section, we illustrate the application of our pricing model to different market dynamics.
Since the continuous-time pricing dynamics does not depend on the choice of the particular model
apart from its covariance structure, we can use simple models for the approximation. For instance,
the Bachelier model (or the Black-Scholes model in the geometric case), is obtained as scaling 
limit of the binomial model. Similarly, we can start from indifference prices defined in a 
binomial model with or without volatility uncertainty and recursively solve the optimal investing 
problem for every trading period. We further illustrate the dependence of the pricing dynamics
on both the level of risk aversion and of uncertainty as well as the convergence of the risk-based 
prices to the worst-case prices as the risk aversion tends to infinity. We also compare different 
linear models with the same covariance structures leading to the same continuous-time pricing 
dynamics although the error bounds might be different. Finally, we observe that the bid-ask spread
for the risk-based prices is clearly smaller the one for the worst-case ones.

Throughout this section, we focus on pricing a butterfly option written on a single asset. Recall 
that a butterfly option with lower strike $K_{\rm L}$, middle strike $K_{\rm M}$ and upper strike 
$K_{\rm U}>0$ gives the holder of the contract the right to obtain at maturity the payoff
\[ f(x)=(x-K_{\rm L})^+ -2(x-K_{\rm M})^+ +(x-K_{\rm U})^+. \]
Usually, one requires $K_{\rm U}-K_{\rm M}=K_{\rm M}-K_{\rm L}$.

In order to produce the numerical illustrations, we always implement\footnote{Source code and examples are available at \url{https://github.com/sgarale/risk_based_pricing}.} the discrete-time approximation given by equation~\eqref{eq:price.limit}.
Note that one could also exploit the characterization of the pricing dynamics as a non-linear PDE
but working on the level of the generator poses additional difficulties when dealing with non smooth 
functions as it is mostly the case in financial contracts. Hence, we directly instead compute at 
each step of the iteration the trading strategy that optimally reduces the risk for the seller 
of the contract in the next trading period.\ We will only consider models satisfying the
conditions~\eqref{eq:non-deg2}-\eqref{eq:strict_ell} which allows us to choose $\Theta:=\Rd$ 
and to consider the limit $\alpha\to\infty$. In particular, we can perform an unconstrained optimization.

\subsection{Implementation}

For the iteration of the one-step operators, we have to find a suitable numerical representation 
of the resulting functions. Starting with a payoff function~$f$, which is known on its entire domain, 
we numerically compute the quantity $(I(t/n)f)(x_i)$ on a finite set $\{x_i\}_{i=1,\dots,N}$.
In order to extend $I(t/n)f$ to its entire domain we then have to prescribe an 
interpolation method. The available possibilities include the following:
\begin{itemize}
 \item directly interpolate $I(t/n)f$, e.g., linearly, using splines, etc,
 \item save the optimizers corresponding to the points $\{x_i\}_{i=1,\dots, N}$ and interpolate 
  them when computing $I(t/n)f$ on new points.
\end{itemize}
 When testing these methods, the first one does not seem feasible: in order to obtain a good 
 approximation of the value $I(t/n)f$, which is then used for the next step of the iteration, 
 one has to start from a very fine spatial grid. This is computationally expensive and 
 becomes even more challenging in higher dimensions. We therefore choose the second option:
 first, we compute $I(t/n)f$ on a set of points covering the region of the domain we are interested in. 
 The resulting optimizers are then interpolated to obtain a better approximation  of $I(t/n)f$
 on a finer grid.
 Furthermore, when the increments of the model are bounded, we can explicitly choose the bounds 
 of the grid to guarantee that errors coming from the part of the domain that we disregard are avoided.
 For the sake of illustration, we consider a one-dimensional binomial model with volatility $\sigma>0$ 
 and suppose that we are interested in approximating the value $S(t)f$ on the interval $[\underline{x},\overline{x}]$
 by an $n$-step iteration with step-size $h:=t/n$. Then, for the last step, the function $I(h)^{n-1}f$ 
 coming from the $(n-1)$-th step, will be evaluated on the region 
 $[\underline{x}-|\mu|h-\sigma\sqrt{h},\overline{x}+|\mu|h+\sigma\sqrt{h}]$. 
 Proceeding backwards, we obtain that it is sufficient to start the iteration with a grid contained 
 in the interval
 \[ [\underline{x}-|\mu|t-\sigma\sqrt{nt},\overline{x}+|\mu|t+\sigma\sqrt{nt}] \]
 in order to avoid errors that might otherwise propagate to the interval $[\underline{x},\overline{x}]$.
 This also shows that a finer time discretization comes at the cost of enlarging the spatial grid.

\subsection{Binomial model} \label{sec:numerics.bin}

We consider a one-dimensional binomial model with drift $\mu\in\R$ and volatility $\sigma>0$,
i.e., the distribution of $\zeta$ under the market expectation is given by
\[ \cE[f(\zeta)]:=\E_\P[f(\zeta)]=\frac{1}{2}\big(f(\sigma)+f(-\sigma )\big) \quad\mbox{for all } f\in\Cb, \]
where $\P\circ\zeta^{-1}:=\frac{1}{2}(\delta_\sigma+\delta_{-\sigma})$.\
Assumption~\ref{ass:main} and condition~\eqref{eq:strict_ell} are clearly satisfied.
Hence, for $\Theta:=\Rd$ and any risk aversion $\alpha>0$, Theorem~\ref{thm:unbounded} yields 
a strongly continuous convex monotone semigroup $(S(t))_{t\geq 0}$ on $\Cb$ which is 
uniquely determined by its generator satisfying $\Cb^2\subset D(A)$ and 
\[ (Af)(x)=\inf_{\theta\in\R}\bigg(\frac{1}{2}\E_\P\big[f''(x)\zeta^2+\alpha\theta^2\zeta^2\big]-\theta\mu\bigg) 
    -\inf_{\theta\in\R}\bigg(\frac{\alpha}{2}\E_\P\big[\theta^2\zeta^2\big]-\theta\mu\bigg) 
    =\frac{1}{2}\sigma^2 f''(x) \]
for all $f\in\Cb^2$ and $x\in\Rd$. Due to Theorem~\ref{thm:unique}, the semigroup $(S(t))_{t\geq 0}$ 
coincides with the linear heat semigroup corresponding to the pricing dynamics under the Bachelier 
model~\cite{Bachelier1900}.
The risk aversion parameter $\alpha>0$ does not appear in the generator of $(S(t))_{t\geq 0}$ 
which is not surprising since the binomial model is complete. Hence, there is no reason for 
the prices to be sensitive to the risk aversion of an agent if the agent can replicate any payoff.

\begin{figure}[h]
 \centering
 \includegraphics[width=0.5\textwidth]{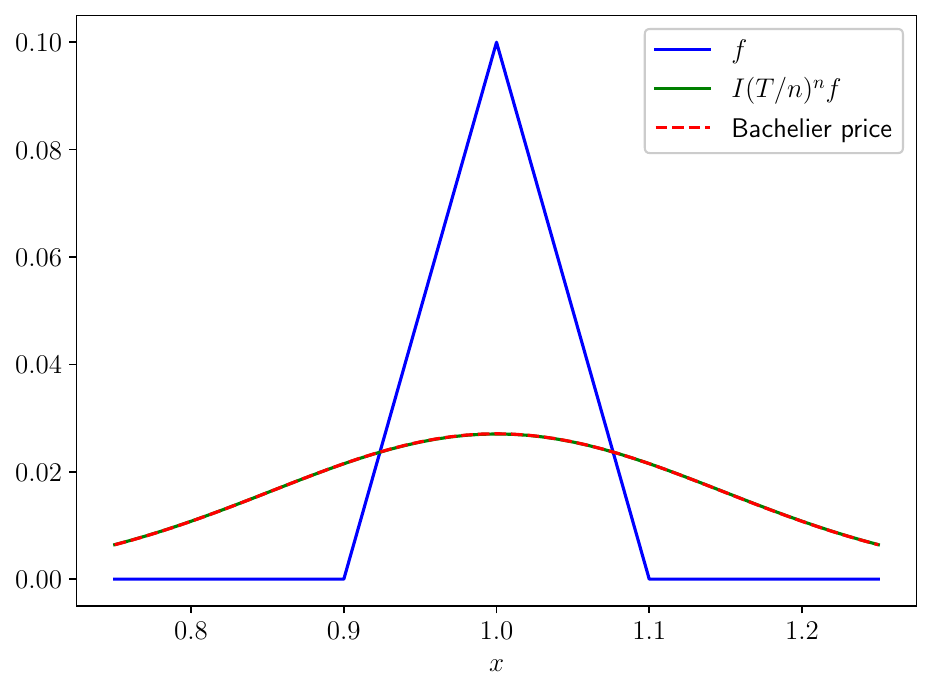}
 \caption{Convergence of the binomial model to the Bachelier model for a butterfly
  option\textsuperscript{\ref{fn:butterfly}}. Maturity: $T=0.5$; number of time steps: $n=200$; 
  parameters of the process: $\sigma=20\%$, $\mu=5\%$; risk aversion: $\alpha=1$.}
 \label{fig:bf_bin}
\end{figure}

Figure~\ref{fig:bf_bin} shows the convergence of the binomial models to the Bachelier 
model for a butterfly option\footnote{Throughout Section~\ref{sec:numerics}, we consider a butterfly 
option with strikes $K_{\rm L} = 0.9$, $K_{\rm M} = 1$ and $K_{\rm U} = 1.1$.\label{fn:butterfly}} 
with maturity $T=0.5$ (6 months) starting from a binomial model with volatility $\sigma=20\%$ 
and drift $\mu=5\%$.

\subsection{Several linear models}
\label{sec:num.linear}

The observation that the risk aversion parameter does not affect the pricing dynamics extends
to any linear model. Indeed, for $\Theta:=\Rd$ and any linear expectation $\cE[\cdot]:=\E_\P[\cdot]$, 
Corollary~\ref{cor:Rd} guarantees that the generator is given by 
\[ (Af)(x)=\frac{1}{2}\E_\P\big[\zeta^T D^2f(x)\zeta\big] 
    \quad\mbox{for all } f\in\Cb^2 \mbox{ and } x\in\Rd. \]
Furthermore, as long as the models share the same covariance structure given by the function 
$G(a):=\E[\zeta^T a\zeta]$ for all $a\in\R^{d\times d}$, the semigroup $(S(t))_{t\geq 0}$
does not depend on the particular choice of the distribution $\P\circ\zeta^{-1}$.

\begin{figure}[h]
 \centering
 \subfloat[Pricing dynamics]{\includegraphics[width=0.45\textwidth]{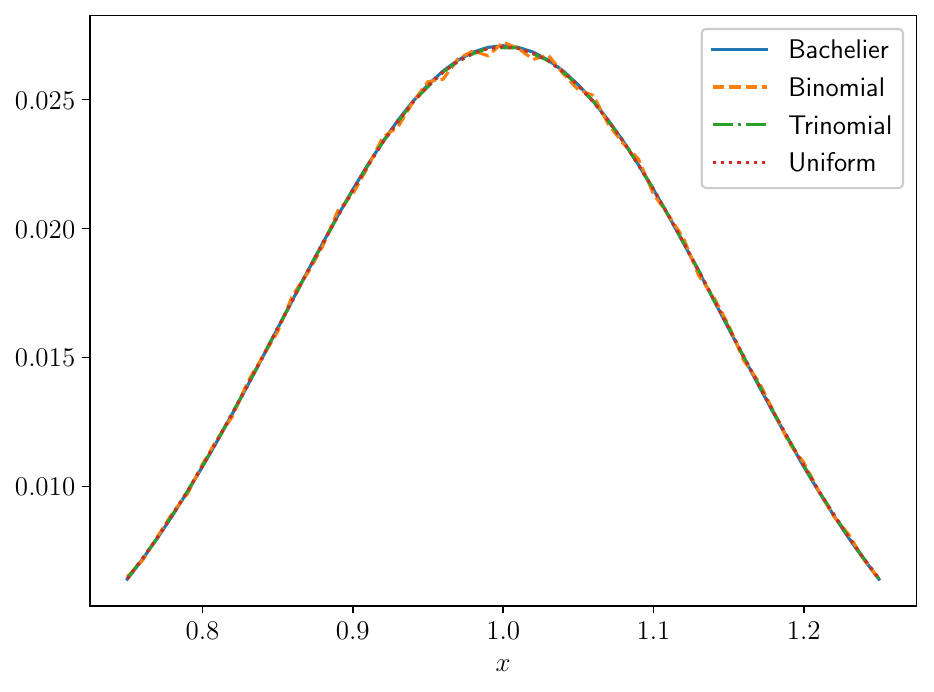}}
 \hfill
 \subfloat[Implied Bachelier volatilities]{\includegraphics[width=0.45\textwidth]{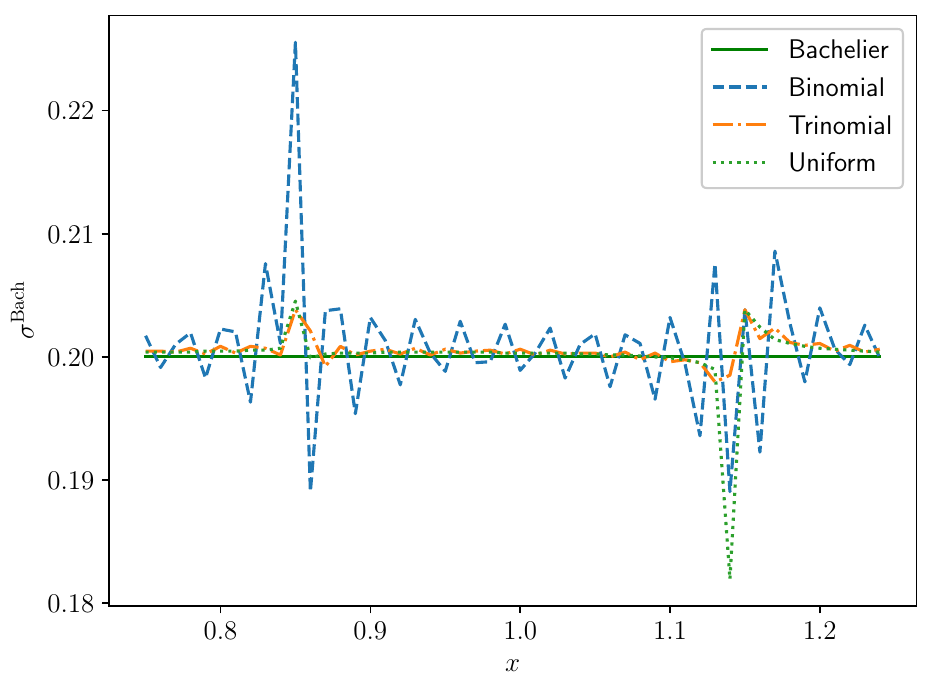}}
 \caption{Convergence of several linear models to the Bachelier model for a butterfly 
  option\textsuperscript{\ref{fn:butterfly}}; Parameters: $n=100$, $\sigma=20\%$, $\mu=5\%$, $\alpha=1$.
  Figure (a) displays the pricing functional for the maturity $T=0.5$; figure (b) shows the 
  corresponding Bachelier implied volatilities.}
 \label{fig:bf_diff_mod}
\end{figure}

Figure~\ref{fig:bf_diff_mod} shows the approximation of the same risk-based prices with
different linear models having the same volatility ($\sigma=20\%$). We compare the following models:
\begin{align*}
 \text{(Binomial)} &\qquad\E_\P[f(\zeta)]=\frac{1}{2}\big(f(\sigma)+f(-\sigma)\big), \\
 \text{(Trinomial)} &\qquad\E_\P[f(\zeta)]=\frac{1}{3}\big(f(\sqrt{3/2}\sigma)+f(0)+f(-\sqrt{3/2}\sigma)\big), \\
 \text{(Uniform)} &\qquad\E_\P[f(\zeta)]=\frac{1}{2\sigma\sqrt{3}}\int_{-\sigma\sqrt{3}}^{\sigma\sqrt{3}} f(x)\,\d x.
\end{align*}
While all these models converge to the same Bachelier price, it seems that richer models 
convergence faster when using the same number of steps in the time discretization. 
This is particularly evident from the plot of the Bachelier implied volatilities.

\subsection{Uncertain binomial model} 
\label{sec:num.unc.bin}

We consider a one-dimensional binomial model with drift $\mu\in\R$ and volatility uncertainty.
Using the parametrization from Example~\ref{ex:sub.expect}(ii), we choose $\nu:=\delta_0$ 
and $\Lambda:=[\sigma_0-u,\sigma_0+u]$, where $\sigma_0>0$ is a reference volatility and $u\in [0,\sigma_0]$ 
is the level of uncertainty. This yields the sublinear expectation
\[ \cE[f(\zeta)]:=\sup_{\sigma\in [\sigma_0-u,\sigma_0+u]}\frac{1}{2}\big(f(\sigma)+f(-\sigma)\big)
    \quad\mbox{for all } f\in\Cb. \]
Here, the risk averse agent fears the worst-case and increases the price as it can be seen 
in Figure~\ref{fig:bf_mult_unc}.

\begin{figure}[h]
 \centering
 \includegraphics[width=0.5\textwidth]{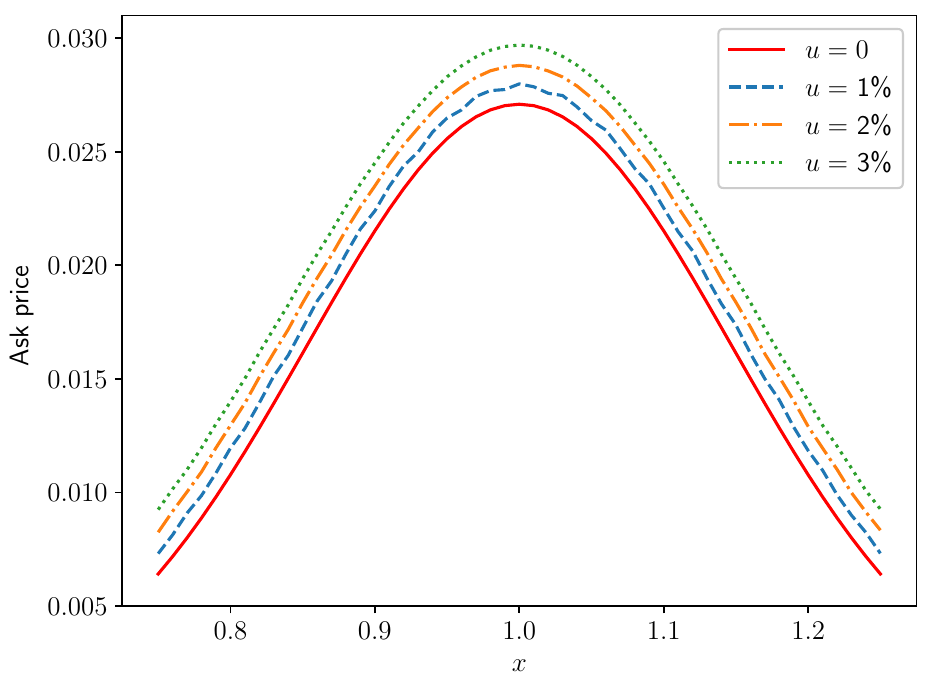}
 \caption{Impact of the level of uncertainty on the risk-based ask price for a butterfly
  option\textsuperscript{\ref{fn:butterfly}}. Parameters: $T=0.5$, $n=100$, $\mu=5\%$, 
  $\sigma_0=20\%$, $\alpha=1$.}
 \label{fig:bf_mult_unc}
\end{figure}

Moreover, condition~\eqref{eq:strict_ell} is satisfied for any $u\in [0,\sigma_0)$ in which case
Theorem~\ref{thm:alpha} implies that the risk-based prices convergence to the $G$-expectation as 
$\alpha\to\infty$. 
Figure~\ref{fig:bf_unc_a} displays the worst-case bound given by the $G$-expectation 
and the risk-based prices for different levels of risk aversion. As shown in Corollary~\ref{cor:Rd},
the risk-based prices are always lower than the worst-case ones. Figure \ref{fig:bf_unc_b} shows 
more in detail how the risk-based prices approach the $G$-expectation as the risk aversion parameter 
increases. Recall from Subsection~\ref{sec:prel.chernoff} that the $G$-expectation is obtained 
by a Chernoff-type approximation with one-step operator $(J(t)f)(x):=\cE[f(x+\sqrt{t}\zeta)]$.

\begin{figure}[h]
 \centering
 \subfloat[Payoff and ask prices\label{fig:bf_unc_a}]{\includegraphics[width=0.45\textwidth]{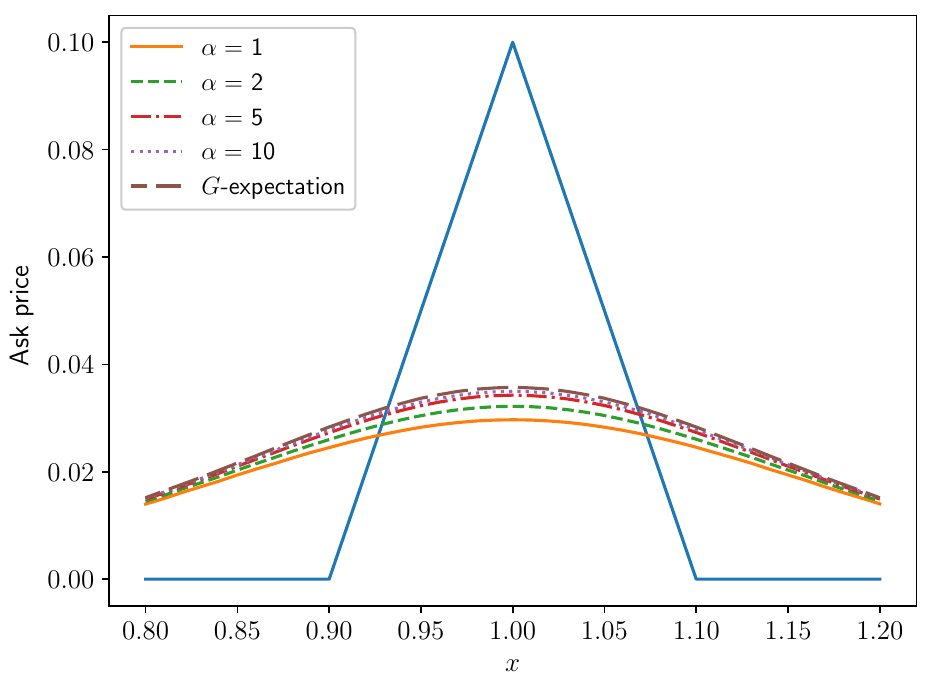}}
 \hfill
 \subfloat[Focus on the risk aversion\label{fig:bf_unc_b}]{\includegraphics[width=0.45\textwidth]{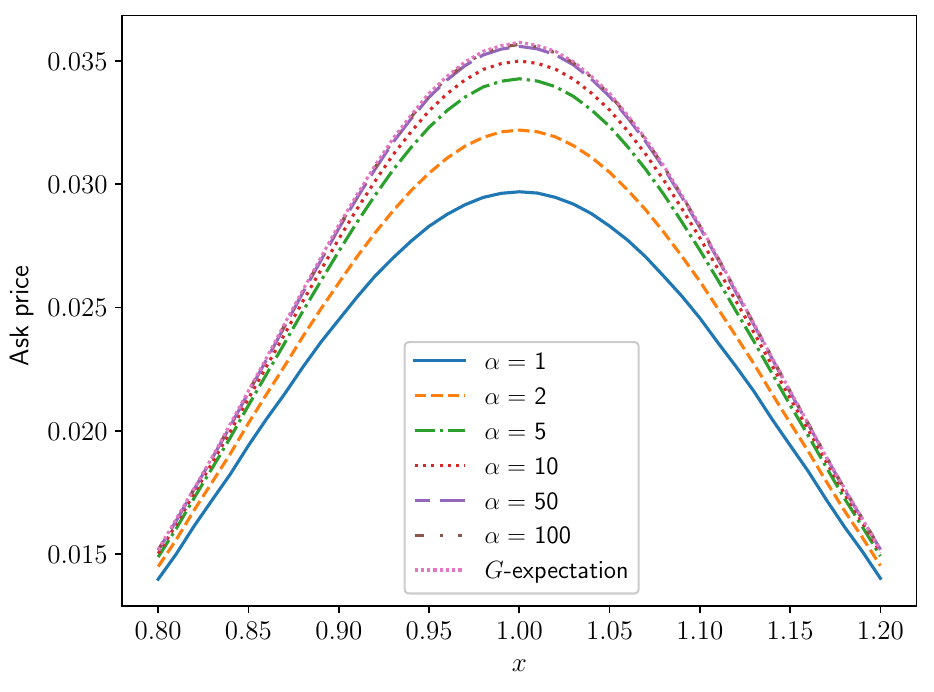}}
 \caption{Impact of the risk aversion parameter on the risk-based ask price for a butterfly
  option\textsuperscript{\ref{fn:butterfly}} and comparison with the worst-case bound.
  Parameters: $T=0.5$, $n=100$, $\mu=5\%$, $\sigma_0=20\%$, $u=3\%$. Figure (a) also 
  displays the payoff function while figure (b) shows more levels of risk aversion.}
 \label{fig:bf_unc_bin}
\end{figure}

\subsection{Bid-ask spread} 
\label{sec:bid-ask}

Using the same arguments as in Section~\ref{sec:indifference}, we can additionally define 
bid pricing operators. Indeed, switching to the buyer position, we obtain that the 
one-step bid prices $b_{h,X_t^h}(f)$ have to satisfy the indifference pricing relation
\[ \inf_{\theta\in\Theta}\rho\big[b_{h,X_t^h}(f)-f(X_{t+h}^x)-\theta^T\big(X_{t+h}^x-X_t^x\big)\big]
    =\inf_{\theta\in\Theta}\rho\big[-\theta^T\big(X_{t+h}^x-X_t^x\big)\big]. \]
Hence, similarly to equation~\eqref{eq:def_I}, we can define the one-step pricing operators
\[ (J(t)f)(x):=b_{t,x}(f)=\tilde{\rho}_{t,x}[0]-\tilde{\rho}_{t,x}[-f]=-\big(I(t)(-f)\big)(x) \]
and the corresponding time consistent multi-step pricing operators 
\[ b_{kh,x}^k:=(J(h)^k f)(x)=-\big(I(h)^k (-f)\big)(x)=-a_{kh,x}^k (-f) \]
for all $f\in\Cb$, $x\in\Rd$ and $k\in\N$. Consequently, for every $t\geq 0$ and $h_n:=t/n$,
the limit 
\[ b_{t,x}^\infty (f):=\lim_{n\to\infty}b_{nh_n,x}^n (f)
    =\lim_{n\to\infty}-a_{nh_n,x}^n (-f)=-a_{t,x}^\infty (-f) \]
exists and satisfies $b_{t,x}^\infty(f)\leq a_{t,x}^\infty(f)$ for all $f\in\Cb$ and $x\in\Rd$.  
Hence, in the absence of trading constraints, Corollary~\ref{cor:Rd} implies that the bid-ask 
spread for the risk-based prices is smaller than the bid-ask spread for the worst-case prices 
associated to the $G$-expectation. 
We now consider again the uncertain binomial model from Subsection~\ref{sec:num.unc.bin} with 
reference volatility $\sigma_0=20\%$ and uncertainty level $u=3\%$.

\begin{figure}[h]
 \centering
 \subfloat[Payoff and bid prices\label{fig:bf_bid_unc}]{\includegraphics[width=0.45\textwidth]{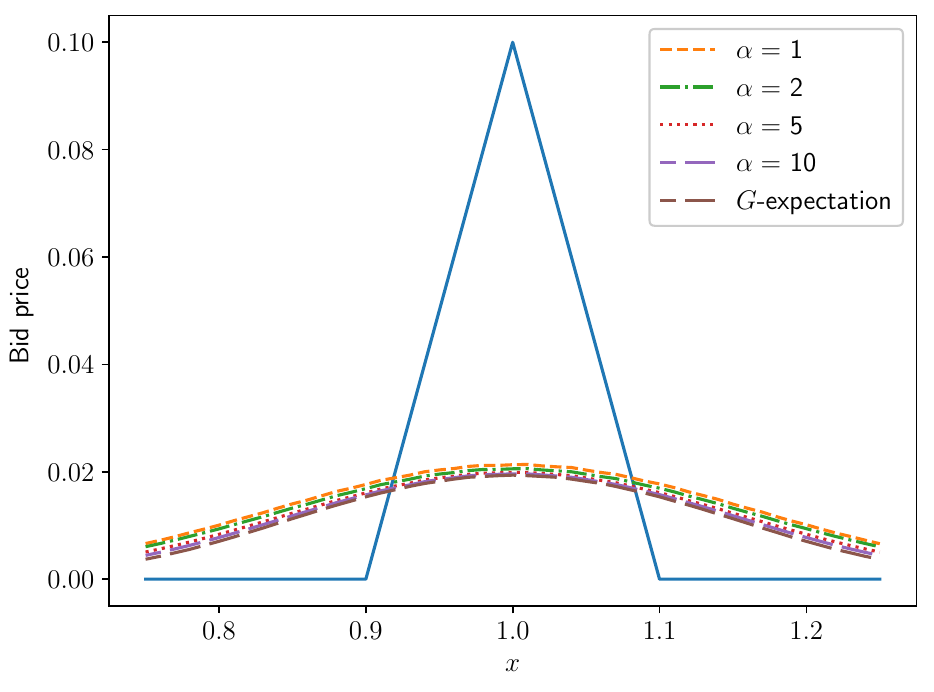}}
 \hfill
 \subfloat[Risk-based bounds vs worst-case bounds\label{fig:bf_bid_ask}]{\includegraphics[width=0.45\textwidth]{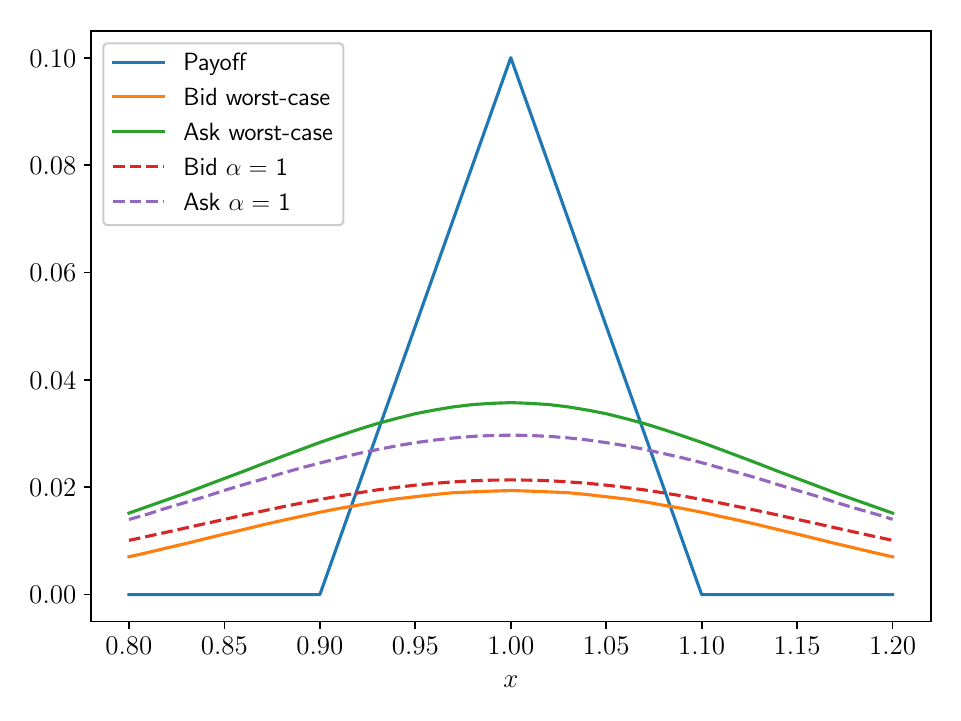}}
 \caption{(a) Impact of the risk aversion on the risk-based bid price for a butterfly
  option\textsuperscript{\ref{fn:butterfly}}. (b) Comparison of the risk-based bid-ask bounds
  with the worst-case bid-ask bounds. Parameters: $T=0.5$, $n=100$, $\mu=5\%$, $\sigma_0=20\%$, $u=3\%$.}
\end{figure}

Figure~\ref{fig:bf_bid_unc} 
shows the bid prices at different levels of risk aversion for a butterfly option\textsuperscript{\ref{fn:butterfly}} 
while Figure~\ref{fig:bf_bid_ask} compares the risk-based bid-ask bounds with the worst-case ones
corresponding to pricing with a $G$-expectation. The risk-based approach, which considers the 
attitude of the agent towards uncertainty, clearly leads to a reduction of the bid-ask spread.

\section{Proof of the main results}
\label{sec:proof}

In order to prove Theorem~\ref{thm:main} and Theorem~\ref{thm:unbounded}, we rely on the following 
results from~\cite{BDKN22}.

\begin{definition}
 Let $(I(t))_{t\geq 0}$ be a family of operators $I(t)\colon\Cb\to\Cb$. The Lipschitz set $\LI$ 
 consists of all $f\in\Cb$ such that there exist $c\geq 0$ and $t_0>0$ with
 \[ \|I(t)f-f\|_\infty\leq ct \quad\mbox{for all } t\in [0,t_0]. \]
 Moreover, for every $f\in\Cb$ such that the following limit exists, we define 
 \[ I'(0)f:=\lim_{h\downarrow 0}\frac{I(h)f-f}{h}\in\Cb. \]
\end{definition}

Recall that all limits in $\Cb$ are taken w.r.t. the mixed topology.

\begin{assumption} \label{ass:cher}
 Let $(I(t))_{t\geq 0}$ be a family of operators $I(t)\colon\Cb\to\Cb$ satisfying 
 the following conditions:
 \begin{enumerate}
  \item $I(0)=\id_{\Cb}$.
  \item $I(t)$ is convex and monotone with $I(t)0=0$ for all $t\geq 0$.
  \item There exists $\omega\geq 0$ with
   \[ \|I(t)f-I(t)g\|_\infty\leq e^{\omega t}\|f-g\|_\infty
   	\quad\mbox{for all } t\in [0,1] \mbox{ and } f,g\in\Cb. \]
  \item For every $\eps>0$, there exist $t_0>0$ and $\delta>0$ with
   \[ I(t)(\tau_x f)\leq\tau_x I(t)f+r\eps t \]
   for all $t\in [0,t_0]$, $x\in B_{\Rd}(\delta)$, $r\geq 0$ and $f\in\Lipb(r)$.
  \item The limit $I'(0)f\in\Cb$ exists for all $f\in\Cbi$.
  \item For every $T\geq 0$, $K\Subset\Rd$ and $(f_k)_{k\in\N}\subset\Cb$ with
   $f_k\downarrow 0$,
   \[ \sup_{(t,x)\in [0,T]\times K}\sup_{n\in\N}\big(I\big(\tfrac{t}{n}\big)^n f_k\big)(x)\downarrow 0
   	\quad\mbox{as } k\to \infty. \]
   \item It holds $I(t)\colon\Lipb(r)\to\Lipb(e^{\omega t}r)$ for all $r,t\geq 0$.
 \end{enumerate}
\end{assumption}

The previous conditions guarantee that~\cite[Assumption~5.7]{BDKN22} is satisfied since 
condition~(vi) is equivalent to~\cite[Assumption~5.7(vi)]{BDKN22}, see~\cite[Lemma~C.2]{BDKN22}.
Hence, the next theorem follows immediately from~\cite[Theorem~5.4, Theorem~5.9 and Corollary~5.11]{BDKN22}.

\begin{theorem} \label{thm:cher}
 Let $(I(t))_{t\geq 0}$ be a family of operators $I(t)\colon\Cb\to\Cb$ satisfying Assumption~\ref{ass:cher}. 
 Then, there exists a strongly continuous convex monotone semigroup $(S(t))_{t\geq 0}$ on $\Cb$ 
 which is given by
 \begin{equation} \label{eq:cher}
  S(t)f:=\lim_{n\to\infty}I\big(\tfrac{t}{n}\big)^n f \quad\mbox{for all } t\geq 0 \mbox{ and } f\in\Cb.
 \end{equation}
 In addition, the following statements are valid:
 \begin{enumerate}
  \item It holds $f\in D(A)$ and $Af=I'(0)f$ for all $f\in\Cb$ such that $I'(0)f\in\Cb$ exists.
   In particular, this is valid for all $f\in\Cbi$.
  \item It holds $\|S(t)f-S(t)g\|_\infty\leq e^{\omega t}\|f-g\|_\infty$ for all $t\geq 0$
   and $f,g\in\Cb$.  
  \item For every $\eps>0$, $r,T\geq 0$ and $K\Subset\R^d$, there exist $K'\Subset\R^d$
   and $c\geq 0$ with
   \[ \|S(t)f-S(t)g\|_{\infty,K}\leq c\|f-g\|_{\infty,K'}+\eps \]
   for all $t\in [0,T]$ and $f,g\in B_{\Cb}(r)$.
  \item It holds $\LI\subset\LS$ and $S(t)\colon\LS\to\LS$ for all $t\geq 0$.
  \item  For every $\eps>0$, there exists $\delta>0$ with  
   \[ S(t)(\tau_x f)\leq\tau_x S(t)f+e^{\omega t}r\eps t \]
  for all $r,t\geq 0$, $f\in\Lipb(r)$ and $x\in B_{\R^d}(\delta)$.
  \item It holds $S(t)\colon\Lipb(r)\to\Lipb(e^{\omega t}r)$ for all $r,t\geq 0$.
 \end{enumerate}
\end{theorem}

It follows from~\cite[Theorem~4.7]{BDKN22} that semigroups which have been constructed this 
way are uniquely determined by their generators evaluated at smooth functions.

\begin{theorem} \label{thm:unique}
 Let $(S(t))_{t\geq 0}$ and $(T(t))_{t\geq 0}$ be two strongly continuous convex monotone
 semigroups on $\Cb$ with generators $A$ and $B$, respectively, which satisfy the conditions~(v) 
 and~(vi) of Theorem~\ref{thm:cher}. Furthermore, we assume that $\Cbi\subset D(A)\cap D(B)$ and 
 \[ Af\leq Bf \quad\mbox{for all } f\in\Cbi. \]
 Then, it holds $S(t)f\leq T(t)f$ for all $t\geq 0$ and $f\in\Cb$. 
\end{theorem}

\subsection{Proof of Theorem~\ref{thm:main}}
\label{sec:proof1}

Recall that
\[ (I_R(t)f)(x):=\inf_{\theta\in\Theta_R}\rho[f(x+\zeta_t)-\theta^T\zeta_t]-c_R(t), \]
for all $R,t\geq 0$, $f\in\Cb$ and $x\in\Rd$, where $\zeta_t:=t\mu+\sqrt{t}\zeta$ and 
$c_R(t):=\inf_{\theta\in\Theta_R}\rho[-\theta^T\zeta_t]$. In the sequel, we fix $R\geq 0$ 
and therefore simply write $I(t)f:=I_R(t)f$, $S(t)f:=S_R(t)f$, $Af:=A_R f$, $\Theta:=\Theta_R$
and $c(t):=c_R(t)$.

\begin{proof}[Proof of Theorem~\ref{thm:main}]
 In order to apply Theorem~\ref{thm:cher}, we have to verify Assumption~\ref{ass:cher}. 
 Furthermore, we have to show that 
 \[ (I'(0)f)(x)=\inf_{\theta\in\Theta}G_\theta(D^2f(x),Df(x))-\inf_{\theta\in\Theta}G_\theta(0,0)
    \quad\mbox{for all } f\in\Cb^2 \mbox{ and } x\in\Rd, \]
 where $G_\theta$ is defined by equation~\eqref{eq:Gtheta}.\ 
 First, we verify Assumption~\ref{ass:cher}(i)-(iv) and~(vii). Condition~(i) follows from 
 $\zeta_0\equiv 0$. Regarding condition~(ii), the monotonicity of $\rho$ yields the monotonicity
 of $I(t)$ and $I(t)0=0$ holds by definition. Let $\lambda\in [0,1]$, 
 $f_1, f_2\in\Cb$, $x\in\Rd$ and $\eps>0$. Let $\theta_1,\theta_2\in\Theta$ be $\eps$-optimizers
 of $(I(t)f_1)(x)$ and $(I(t)f_2)(x)$, respectively, and define $\tilde{f}=\lambda f_1+(1-\lambda)f_2\in\Cb$ 
 and $\tilde{\theta}=\lambda\theta_1+(1-\lambda)\theta_2\in\Theta$. Then, by convexity of $\rho$, 
 \begin{align*}
  \big(I(t)\tilde{f}\big)(x) 
  &\leq\rho \big[\tilde{f}(x+\zeta_t) + \tilde{\theta}^T\zeta_t \big]-c(t) \\
  &\leq\lambda\big(\rho\big[f_1(x+\zeta_t)+\theta_1^T\zeta_t\big]-c(t)\big)
    +(1-\lambda)\big(\rho\big[f_2(x+\zeta_t)+\theta_2^T\zeta_t\big]-c(t)\big) \\
  &=\lambda(I(t)f_1)(x)+(1-\lambda)(I(t)f_2)(x)+\eps.
 \end{align*}
 Letting $\eps\downarrow 0$ shows that $I(t)$ is convex. For every $t\geq 0$, $f,g\in\Cb$, 
 $x\in\Rd$ and $\theta\in\Theta$, the monotonicity and cash invariance of $\rho$ imply 
 \[ (I(t)f)(x)\leq\rho\big[f(x+\zeta_t)+\theta^T\zeta_t\big]-c(t)
  \leq\rho\big[g(x+\zeta_t)+\theta^T\zeta_t\big]-c(t)+\|f-g\|_\infty. \]
 Taking the infimum over $\theta\in\Theta$ and changing the roles of $f$ and $g$ yields
 \[ \|I(t)f-I(t)g\|_\infty\leq\|f - g\|_\infty \quad\mbox{for all } t\geq 0 \]
 which shows that condition~(iii) is satisfied with $\omega:=0$. 
 Condition~(iv) follows from 
 \[ I(t)(\tau_x f)=\tau_x I(t)f \quad\mbox{for all } t\geq 0, f\in\Cb  \mbox{ and } x\in\Rd. \]
 In order to verify condition~(vii) with $\omega:=0$, we observe that 
 \begin{align*}
  |(I(t)f)(x)-(I(t)f)(y)| &=|(\tau_x I(t)f)(0)-(\tau_y I(t)f)(0)|
    =|(I(t)\tau_x f)(0)-(I(t)\tau_y f)(0)| \\
  &\leq\|\tau_x f-\tau_y f\|_\infty\leq r|x-y|
 \end{align*}
 for all $r,t\geq 0$, $f \in \Lip_b(r)$ and $x\in\Rd$.

 Second, we show that $\Cb^2\subset\LI$. Let $f\in\Cb^2$. For every $t\in [0,1]$, $x\in\Rd$ 
 and $\theta\in\Theta$, applying Taylor's formula on the function
 $g(y):=\exp(\alpha (f(x+y)-f(x))-\alpha\theta^T y)$ yields  
 \begin{align} 
  &\exp\big(\alpha(f(x+\zeta_t)-f(x))-\alpha\theta^T\zeta_t\big)
    =1+\alpha (Df(x)-\theta)^T\zeta_t \nonumber \\
  &\quad\; +\int_0^1\big(\alpha^2|(Df(x+s\zeta_t)-\theta)^T\zeta_t|^2+\alpha\zeta_t^T D^2f(x+s\zeta_t)\zeta_t\big)
    g(s\zeta_t)(1-s)\,\d s. \label{eq:taylor}
 \end{align}
 It follows from Assumption~\ref{ass:main} and Lemma \ref{lem:E}(vi) that
 \begin{align*}
  &\big|\cE\big[\exp(\alpha (f(x+\zeta_t)-f(x))-\alpha\theta^T\zeta_t)\big]-1\big| \\
  &\leq\alpha (\|Df\|_\infty+R)|\mu|t
     +\frac{\alpha}{2}\cE\left[\big(\alpha(\|Df\|_\infty+R)^2+\|D^2f\|_\infty\big)
     |\zeta_t|^2 e^{2\alpha\|f\|_\infty+\alpha R|\zeta_t|}\right] \\
  &\leq\alpha(\|Df\|_\infty+R)|\mu|t \\
  &\quad\; +\frac{\alpha}{2}\cE\left[\big(\alpha(\|Df\|_\infty+R)^2+\|D^2f\|_\infty\big)
    \big(|\mu|^2+2|\mu|\cdot |\zeta|+|\zeta|^2\big)e^{\alpha(2\|f\|_\infty+R(|\mu|+|\zeta|))}\right]t
 \end{align*}
 for all $t\in [0,1]$, $x\in\Rd$ and $\theta\in\Theta$. Hence, there exist $c_1,c_2\geq 0$
 and $t_0>$ with
 \[ -c_2 t\leq\log(1-c_1 t)
    \leq\log\big(\cE\big[\exp(\alpha f(x+\zeta_t)-\alpha\theta^T\zeta_t\big]\big)-f(x)
    \leq\log(1+c_1 t)\leq c_2t \]
 for all $t\in [0,t_0]$, $x\in\Rd$ and $\theta\in\Theta$. Dividing by $\alpha>0$, taking the supremum 
 over $\theta\in\Theta$ and applying the previous estimate with $f\equiv 0$ shows that
 \[ \|I(t)f-f\|_\infty
    \leq\sup_{\theta\in\Theta}\big\|\rho\big[f(\,\cdot\,+\zeta_t)-\theta^T\zeta_t\big]-f\big\|_\infty
    +\sup_{\theta\in\Theta}\big|\rho\big[-\theta^T\zeta_t\big]\big|\leq 2c_2t 
    \quad\mbox{for all } t\in [0,t_0]. \]

 Third, we verify condition~(vi).\ Several sufficient conditions which guarantee that 
 condition~(vi) is valid have been systemically explored in~\cite[Subsection~2.5]{BKN23}. 
 Here, we show that, for every $r\geq 0$, there exists $c\geq 0$ with
 \begin{equation} \label{eq:cont.above}
  \|I(t)f-f\|_\infty\leq c\big(\|Df\|_\infty+\|D^2f\|_\infty\big)t
 \end{equation}
 for all $t\in [0,1]$ and $f\in\Cb^2$ with $\|f\|_\infty\leq r$ and $\|Df\|_\infty+\|D^2f\|_\infty\leq 1$.
 Let $r\geq 0$, $t\in [0,1]$ and $f\in\Cb^2$ with $\|f\|_\infty\leq r$ and $0<\|Df\|_\infty+\|D^2f\|_\infty\leq 1$.
 For every $\lambda\in (0,1)$, $x\in\Rd$ and $\theta\in\Theta$, applying equation~\eqref{eq:taylor} 
 with $\frac{1}{\lambda}f$ yields
 \begin{align*}
  &\cE\left[\exp\left(\frac{\alpha(f(x+\zeta_t)-f(x))}{\lambda}-\alpha\theta^T\zeta_t\right)\right] \\
  &\leq 1+\alpha\big(\tfrac{1}{\lambda}\|Df\|_\infty+R\big)|\mu| t \\
  &\quad\; +\alpha\cE\left[\int_0^1
    \Big(\tfrac{\alpha}{\lambda^2}\|Df\|_\infty^2+\tfrac{\alpha}{\lambda}\|Df\|_\infty R+\alpha R^2
    +\tfrac{1}{\lambda}\|D^2 f\|_\infty\Big)
    |\zeta_t|^2 g(sX_s)(1-s)\,\d s\right].
\end{align*}
 We apply Lemma~\ref{lem:lambda} with $\lambda:=\|Df\|_\infty+\|D^2f\|_\infty\in (0,1]$ to obtain
 \begin{align*}
  &\rho\big[f(x+\zeta_t)-f(x)-\theta^T\zeta_t\big]-\rho\big[-\theta^T\zeta_t\big]
    \leq\lambda\rho\left[\frac{f(x+\zeta_t)-f(x)}{\lambda}-\theta^T\zeta_t\right] 
    -\lambda\rho\big[-\theta^T\zeta_t\big] \\
  &\leq\frac{\lambda}{\alpha}\log\left(1+\alpha (1+R)|\mu|t
    +\frac{\alpha^2 (1+R+R^2)+\alpha)}{2}\cE\left[|\zeta_t|^2 e^{\alpha (2r+R|\zeta_t|)}\right]\right) \\
  & \quad\; +\frac{\lambda}{\alpha}\log\left(1+\alpha R|\mu|t+\frac{\alpha^2 R^2}{2}
    \cE\left[|\zeta_t|^2e^{\alpha R|\zeta_t|}\right]\right) \\
 & \leq\frac{\lambda}{\alpha}\log\big(1+\alpha (1+R)|\mu|t + \cE[Y]t\big)
  +\frac{\lambda}{\alpha}\log\big(1+\alpha R|\mu|t+\cE[Z]t\big)
 \end{align*}
 for all $x\in\Rd$ and $\theta\in\Theta$, where 
 \begin{align*}
  Y &:=\tfrac{1}{2}\big(\alpha^2(1+R+R^2)+\alpha\big)\big(|\mu|^2+2|\mu|\cdot |\zeta|+|\zeta|^2\big)
    e^{\alpha(2r+R(|\mu|+|\zeta|))}, \\
  Z &:=\tfrac{1}{2}\alpha^2R^2\big(|\mu|^2+2|\mu|\cdot |\zeta|+|\zeta|^2\big)e^{\alpha R(|\mu|+|\zeta|)}.
 \end{align*}
 The lower bound follows similarly and therefore taking the supremum over $\theta\in\Theta$ shows 
 that, for every $r\geq 0$, there exists $c\geq 0$ with
 \begin{align*}
  |(I(t)f-f)(x)|
  &=\Big|\inf_{\theta\in\Theta}\rho[f(x+\zeta_t)-\theta^T\zeta_t]-\inf_{\theta\in\Theta}\rho[-\theta^T\zeta_t]\Big| \\
  &\leq\sup_{\theta\in\Theta}\big|\rho\big[f(x+\zeta_t)-f(x)-\theta^T\zeta_t\big]-\rho\big[-\theta^T\zeta_t\big]\big|
    \leq c\big(\|Df\|_\infty+\|D^2 f\|_\infty\big)t 
 \end{align*}
 for all $t\in [0,1]$, $f\in\Cb^2$ with $\|f\|_\infty\leq r$ and $0<\|Df\|_\infty+\|D^2f\|_\infty\leq 1$
 and $x\in\Rd$. Hence, we can apply~\cite[Corollary~2.14]{BKN23} with $S_n(t)f:=I(\frac{t}{n})^n f$, 
 $\cT_n:=\R_+$ and $X_n:=\Rd$ for all $t\geq 0$, $f\in\Cb$ and $n\in\N$ to obtain that condition~(vi) 
 is satisfied\footnote{The result in~\cite{BKN23} is stated under a slightly stronger condition but 
 a close inspection of the proof reveals that it is sufficient to verify inequality~\eqref{eq:cont.above}
 instead of~\cite[Inequality~(2.35)]{BKN23}.}.

 Fourth, for every $f\in\Cb^2$, we show that the limit $I'(0)f\in\Cb$ exists and is given by
\[ (I'(0)f)(x)=\inf_{\theta\in\Theta}G_\theta(D^2f(x),Df(x))-\inf_{\theta\in\Theta}G_\theta(0,0)
    \quad\mbox{for all } x\in\Rd, \]
 where $G_\theta(a,b)=\frac{1}{2}\cE[\zeta^T a\zeta+\alpha |(b-\theta)^T\zeta|^2]+(b-\theta)^T\mu$
 for all $a\in\R^{d\times d}$, $b\in\Rd$ and $\theta\in\Theta$. Let $f\in\Cb^2$.
 Equation~\eqref{eq:taylor}, Assumption~\ref{ass:main} and Lemma~\ref{lem:E}(i) and~(vi) imply
 \begin{align}
  &\frac{1}{\alpha t}\log\big(\cE\big[\exp\big(\alpha (f(x+\zeta_t)-f(x))-\alpha\theta^T\zeta_t\big)\big]\big) 
    - G_\theta(D^2f(x),Df(x)) \nonumber \\
  &=\frac{1}{\alpha t}\log\big(1+\alpha (Df(x)-\theta)^T\mu t+\cE[Y_t^{\theta,x}]\big)
    -(Df(x)-\theta)^T\mu\nonumber \\
  & \quad\; -\frac{1}{2}\cE\left[\alpha|(Df(x)-\theta)^T\zeta|^2+\zeta^T D^2f(x)\zeta\big)\right] \nonumber \\
  &=\frac{1}{\alpha t}\cE[Y_t^{\theta,x}]-\cE[Z^{\theta,x}]+r_t^{\theta,x} \label{eq:gen.terms1}
 \end{align}
 for all $t>0$, $x\in\Rd$ and $\theta\in\Theta$, where
 \begin{align*}
  Y_t^{\theta,x} &:=\int_0^1\big(\alpha^2|(Df(x+s\zeta_t)-\theta)^T\zeta_t|^2
    +\alpha\zeta_t^T D^2f(x+s\zeta_t)\zeta_t\big)g^{\theta,x}(s\zeta_t)(1-s)\,\d s, \\
  Z^{\theta,x} &:=\frac{1}{2}\alpha |(Df(x)-\theta)^T\zeta|^2+\frac{1}{2}\zeta^T D^2f(x)\zeta, \\
  r_t^{\theta,x} &:=\frac{1}{\alpha t}\log\big(1+\alpha (Df(x)-\theta)^T\mu t+\cE[Y_t^{\theta,x}]\big) 
    -\Big((Df(x)-\theta)^T\mu+\frac{1}{\alpha t}\cE[Y_t^{\theta,x}]\big)\Big)
 \end{align*}
 and $g^{\theta,x}\colon\Rd\to\R,\; y\mapsto\exp(\alpha (f(x+y)-f(x))-\alpha\theta^T y)$. 
 For every $t\in (0,1]$, $x\in\Rd$ and $\theta\in\Theta$, we can estimate 
 \[ |\alpha (Df(x)-\theta)^T\mu t|\leq\alpha (\|Df\|_\infty+R)|\mu| t \]
 and, by defining $c_f:=\alpha^2(\|Df\|_\infty+R)^2+\alpha\|D^2f\|_\infty$, 
 \begin{align} \label{eq:Y_t.bound}
  \cE[|Y_t^{\theta,x}|] 
  &\leq\cE\left[\frac{1}{2}\big(\alpha^2(\|Df\|_\infty+R)^2+\alpha\|D^2f\|_\infty\big)
    |\zeta_t|^2e^{\alpha(2\|f\|_\infty+R|\zeta_t|)}\right] \nonumber \\
  &\leq\cE\left[\frac{\alpha}{2}c_f\big(|\mu|^2+2|\mu|\cdot|\zeta|+|\zeta|^2\big)
    e^{\alpha(2\|f\|_\infty+R(|\mu|+|\zeta|))}\right] t. 
 \end{align}
 Since $\log(1+z)=z+o(|z|)$ for $z\to 0$, we obtain 
 \begin{equation} \label{eq:gen.r}
  \sup_{x\in\Rd}\sup_{\theta\in\Theta}|r_t^{\theta,x}|\to 0 \quad\mbox{as } t\downarrow 0.
 \end{equation}
 It remains to show that, for every $K\Subset\Rd$, 
 \begin{equation} \label{eq:gen.conv1}
  \sup_{x\in K}\sup_{\theta\in\Theta}\left|\frac{1}{\alpha t}\cE[Y_t^{\theta,x}]-\cE[Z^{\theta,x}]\right|\to 0
  \quad\mbox{as } t\downarrow 0.
 \end{equation}
 For every $t \in (0,1]$, $x\in\Rd$ and $\theta \in \Theta$, Lemma \ref{lem:E}(viii) implies
 \[ \left|\frac{1}{\alpha t}\cE[Y_t^{\theta,x}]-\cE[Z^{\theta,x}]\right| 
    \leq\cE\left[\left|\frac{1}{\alpha t}Y_t^{\theta,x}-Z^{\theta,x}\right|\right]. \]
 Moreover, for every $\eps>0$, by inequality~\eqref{eq:Y_t.bound} and Assumption~\ref{ass:main},
 there exists $c\geq 0$ with 
 \begin{equation} \label{eq:gen.conv2}
  \cE\left[\left|\frac{1}{\alpha t}Y_t^{\theta,x}-Z^{\theta,x}\right|\one_{\{|\zeta|>c\}}\right] 
  \leq\eps\quad\mbox{for all } t\in (0,1], x\in\Rd \mbox{ and } \theta\in\Theta. 
 \end{equation}
 Let $\eps>0$ and choose $c\geq 0$ such that the previous inequality is valid. It follows from 
 \[ |\zeta_t|^2=(t|\mu|^2 + 2\sqrt{t}\mu^T\zeta + |\zeta|^2)t \]
 that there exists $C\geq 0$ with 
 \begin{align*}
  &\cE\left[\left|\frac{1}{\alpha t}Y_t^{\theta,x}-Z^{\theta,x}\right|\one_{\{|\zeta|\leq c\}}\right] \\
  &\leq\alpha\cE\left[\int_0^1\big(|(Df(x+s\zeta_t)-\theta)^T\zeta|^2
    -|(Df(x)-\theta)^T\zeta|^2\big)\one_{\{\zeta|\leq c\}}(1-s)\,\d s\right] \\
  &\quad\; +\cE\left[\int_0^1\big(|\zeta^T D^2f(x+s\zeta_t)\zeta g^{\theta,x}(s\zeta_t)
    -\zeta^T D^2f(x)\zeta\big)\one_{\{\zeta|\leq c\}}(1-s)\,\d s\right]+C\sqrt{t}
 \end{align*}
 for all $t\in (0,1]$, $x\in\Rd$ and $\theta\in\Theta$. Let $K\Subset\Rd$. 
 Since the functions $\{g^{\theta,x}\colon\theta\in\Theta, x\in K\}$ are equicontinuous
 with $g^{\theta,x}(0)=1$ and $\zeta_t\one_{\{|\zeta|\leq c\}}$ converges uniformly to 
 zero as $t\downarrow 0$, there exists $t_0\in (0,1]$ with 
 \begin{equation} \label{eq:gen.conv3}
  \cE\left[\left|\frac{1}{\alpha t}Y_t^{\theta,x}-Z^{\theta,x}\right|\one_{\{|\zeta|\leq c\}}\right]\leq\eps 
  \quad \mbox{for all } t\in (0,t_0], x\in K \mbox{ and } \theta\in\Theta. 
 \end{equation}
 Combining the inequalities~\eqref{eq:gen.conv2} and~\eqref{eq:gen.conv3} shows that 
 equation~\eqref{eq:gen.conv1} is valid.\ Hence, it follows from equation~\eqref{eq:gen.terms1} that 
 \begin{align*}
  & \sup_{x\in K}\left|\frac{1}{t}\Big(\inf_{\theta\in\Theta}\rho[f(x+\zeta_t)-\theta^T\zeta_t]-f(x)\Big) 
    -\inf_{\theta\in\Theta}G_\theta(D^2f(x),Df(x))\right| \\
  &\leq\sup_{x\in K}\sup_{\theta\in\Theta}\frac{1}{\alpha t}
    \big|\log\big(\cE\big[\exp\big(\alpha (f(x+\zeta_t)-f(x))-\alpha\theta^T\zeta_t\big)\big]\big) 
    - G_\theta(D^2f(x),Df(x))\big| \\
  &\leq\sup_{x\in K}\sup_{\theta\in\Theta}
    \left|\frac{1}{\alpha t}\cE[Y_t^{\theta,x}]-\cE[Z^{\theta,x}]+r_t^{\theta,x}\right|\to 0 
    \quad \mbox{as } t\downarrow 0.
 \end{align*}
 By additionally applying this result on the constant function $f\equiv 0$, we obtain that 
 the limit $I'(0)f\in\Cb$ exists and is given by
 \[ (I'(0)f)(x)=\inf_{\theta\in\Theta}G_\theta(D^2 f(x),Df(x))-\inf_{\theta\in\Theta}G_\theta(0, 0) 
    \quad\mbox{for all } f\in\Cb^2 \mbox{ and } x\in\Rd. \]
 Now, the claim follows from Theorem~\ref{thm:cher} and Theorem~\ref{thm:unique}.
\end{proof}

\subsection{Proof of Theorem~\ref{thm:unbounded}}
\label{sec:proof2}

The proof is based on the stability results for strongly continuous convex monotone semigroups in~\cite{BKN23}.\
The basic idea is that, for a sequence of semigroups satisfying certain stability conditions,
convergence of the generators implies convergence of the corresponding semigroups. In many applications,
the generators can be determined explicitly for smooth functions and the same applies to their convergence. 
In contrast, showing the convergence of the semigroups directly is often not feasible.

\begin{proof}[Proof of Theorem~\ref{thm:unbounded}]
 In order to apply~\cite[Theorem~2.3]{BKN23}, we have to verify~\cite[Assumption~2.2]{BKN23}. 
 The conditions~(i),~(iii) and~(iv) follow from Theorem~\ref{thm:cher} and the corresponding 
 estimates in the proof of Theorem~\ref{thm:main}. In order to verify the conditions~(ii) and~(v), 
 we show that, for every $c\geq 0$, there exist $R\geq 0$ and $t_0>0$ with 
 \begin{equation} \label{eq:unbounded1} 
  \inf_{\theta\in\Theta}G_\theta(D^2 f(x),Df(x))=\inf_{\theta\in\Theta_R}G_\theta(D^2 f(x),Df(x))
 \end{equation}
 for all $t\in [0,t_0]$, $f\in\Cb^2$ with $\|Df\|_\infty+\|D^2 f\|_\infty\leq c$ and $x\in\Rd$.
 Choose $\delta>0$ such that condition~\eqref{eq:non-deg} is satisfied. For every $f\in\Cb^2$, 
 $x\in\Rd$ and $\theta\in\Theta$, Lemma~\ref{lem:E}(vii) implies
 \begin{align*}
  G_\theta(D^2 f(x),Df(x))
  &=\frac{1}{2}\cE\big[\alpha |(Df(x)-\theta)^T\zeta|^2+\zeta^T D^2f(x)\zeta\big]+(Df(x)-\theta)^T\mu \\
  &=\frac{1}{2}\cE\big[\alpha |\theta^T\zeta|^2-2\alpha Df(x)^T\zeta\theta^T\zeta + \alpha|Df(x)^T\zeta|^2+\zeta^T D^2f(x)\zeta\big] \\
  &\quad  +(Df(x)-\theta)^T\mu \\
  &\geq\frac{\alpha\delta}{2}|\theta|^2-c_1|\theta|-c_2, 
 \end{align*}
 where $c_1:=\alpha\|Df\|_\infty\cE[|\zeta|^2]+|\mu|$ and 
 $c_2:=\frac{1}{2}(\alpha\|Df\|_\infty^2+\|D^2 f\|_\infty)\cE[|\zeta|^2]+\|Df\|_\infty|\mu|$. 
 This shows that equation~\eqref{eq:unbounded1} is valid and condition~(v) follows immediately.
 For $\Theta=\Rd$, we observe that, for every $\theta\in\Rd$ with $(Df(x)-\theta)^T\mu=0$, 
 \[ G_\theta(D^2f(x),Df(x))\geq\frac{1}{2}\cE[\zeta^T D^2f(x)\zeta] \]
 with equality for $\theta=Df(x)$. Moreover, for every $\theta\in\Rd$ with $(Df(x)-\theta)^T\mu\neq 0$,
 one can use condition~\eqref{eq:non-deg2} to estimate 
 \[ G_\theta(D^2f(x),Df(x))\geq\frac{\alpha\delta}{2}|Df(x)-\theta|^2-c_1|Df(x)-\theta|-c_2 \]
 for suitable constants $c_1,c_2\geq 0$. Hence, equation~\eqref{eq:unbounded1} is still
 valid for some $R\geq\|Df\|_\infty$. In order to verify condition~(iii), we further show 
 that there exists $c\geq 0$ with 
 \begin{equation} \label{eq:unbounded2}
  \|A_R f\|_\infty\leq c\big(\|Df\|_\infty+\|D^2 f\|_\infty\big)
 \end{equation}
 for all $f\in\Cb^2$ with $\|Df\|_\infty+\|D^2 f\|_\infty\leq 1$. By equation~\eqref{eq:unbounded1},
 there exists $R_0\geq 0$ with 
 \[ |(A_R f)(x)|\leq\sup_{\theta\in\Theta_{R_0}}\big|G_\theta(D^2 f(x),Df(x))-G_\theta(0,0)\big| \]
 for all $R\geq 0$, $f\in\Cb^2$ with $\|Df\|_\infty+\|D^2 f\|_\infty\leq 1$ and $x\in\Rd$. 
 For every $f\in\Cb^2$ with $\|Df\|_\infty+\|D^2 f\|_\infty\leq 1$, $x\in\Rd$ and $\theta\in\Theta_{R_0}$,
 it follows from Lemma \ref{lem:E}(viii) that
 \begin{align*}
  &\big|G_\theta(D^2 f(x),Df(x))-G_\theta(0,0)\big| \\
  &=\left|\frac{1}{2}\cE\big[\zeta^T D^2f(x)\zeta+\alpha |(Df(x)+\theta)^T\zeta|^2\big]+Df(x)^T\mu
    -\frac{\alpha}{2}\cE\left[|\theta^T\zeta|^2\right]\right| \\
  &\leq\frac{1}{2}\cE\big[\zeta^T D^2f(x)\zeta+\alpha|Df(x)^T\zeta|^2+2\alpha |Df(x)|\cdot |\theta|\cdot |\zeta|^2\big]
    +|Df(x)^T\mu| \\
  &\leq\frac{1}{2}\big(\|D^2 f\|_\infty+\alpha (\|Df\|_\infty^2 + 2\|Df\|_\infty |\theta|)\big)\cE[|\zeta|^2]
    +\|Df\|_\infty |\mu|\\
  &\leq\left(\frac{1}{2}(1\vee\alpha)(1 + 2R_0)\cE[|\zeta|^2]+|\mu|\right)\big(\|Df\|_\infty+\|D^2 f\|_\infty\big).
 \end{align*}
 This shows that inequality~\eqref{eq:unbounded2} is satisfied. Furthermore, it holds
 \[ \lim_{h\downarrow 0}\left\|\frac{S_R(h)f-f}{h}-A_R f\right\|_\infty=0
    \quad\mbox{for all } R\geq 0 \mbox{ and } f\in\BUC^2, \]
 where $\BUC^2$ denotes the space of all bounded twice differentiable functions $f\colon\Rd\to\R$ 
 such that the first and second are bounded and uniformly continuous. Indeed, for every $R\geq 0$
 and $f\in\BUC^2$, the supremum in inequality~\eqref{eq:gen.conv3} and therefore in 
 inequality~\eqref{eq:gen.conv1} can be taken over $x\in\Rd$ instead of $x\in K$. We obtain
 \[ \lim_{h\downarrow 0}\left\|\frac{I_R(h)f-f}{h}-I'(0)f\right\|_\infty=0 \]
 and~\cite[Theorem~4.3]{BK23} transfers the previous statement to the semigroup $(S_R(t))_{t\geq 0}$.\
 It follows from~\cite[Corollary~2.16]{BKN23} that condition~(iii) is satisfied for any sequence 
 $R_n\to\infty$\footnote{The result in~\cite{BKN23} is stated under a slightly stronger condition 
 but a close inspection of the proof reveals that it is sufficient to verify inequality~\eqref{eq:unbounded2}.}.
 The claim now follows from~\cite[Theorem~2.5]{BKN23} since Theorem~\ref{thm:unique} guarantees 
 that the limit does not depend on the choice of the sequence $R_n\to\infty$.
\end{proof}

\subsection{Proof of Theorem~\ref{thm:unbounded2}}
\label{sec:proof3}

The proof is similar to the one of Theorem~\ref{thm:main} as soon as we can constrain
the trading strategies to a bounded set.

\begin{proof}
 By condition~\eqref{eq:mgf2}, there exists $t_0>0$ such that $I(t)$ is well-defined for all $t\in [0,t_0]$.
 Hence, since we are only interested in the limit behaviour of the iterated operators, we can replace
 $I(t)$ by $I(t_0)$ for all $t>t_0$ without affecting the result. We subsequently show that Assumption~\ref{ass:cher}
 is satisfied.\ The conditions~(i)-(iv) and~(vii) can be verified exactly as in the proof of Theorem~\ref{thm:main},
 where we did not use that the set $\Theta=\Theta_R$ was bounded.\ In order to verify condition~(v),
 let $f\in\Cb^2$. We show that there exist $R\geq 0$ and $t_0>0$ with
 \begin{equation} \label{eq:bounded}
  \inf_{\theta\in\Rd}\rho[f(x+\zeta_t)-\theta^T\zeta_t]=\inf_{|\theta|\leq R}\rho[f(x+\zeta_t)-\theta^T\zeta_t]
 \end{equation}
 for all $t\in [0,t_0]$ and $x\in\Rd$. For every $t\geq 0$, $x\in\Rd$ and $\theta\in\Theta$,
 Taylor's formula implies
 \[ f(x+\zeta_t)=f(x)+Df(x)^T\zeta_t+\int_0^1\zeta_t^T D^2f(x+\zeta_t)\zeta_t (1-s)\,\d s. \]
 We choose $\theta=Df(x)$ and apply condition~\eqref{eq:mgf1} to obtain $M\geq 0$ and $t_1\in (0,1]$ with 
 \begin{align*}
  \rho[f(x+\zeta_t)-\theta^T\zeta_t]
  &=\frac{1}{\alpha}\log\left(\cE\left[\exp\left(\alpha f(x)
    +\alpha\int_0^1\zeta_t^T D^2f(x+s\zeta_t)\zeta_t (1-s)\,\d s\right)\right]\right) \\
  &\leq f(x)+\frac{1}{\alpha}\log\big(\cE\big[[e^{2 \alpha\|D^2f\|_\infty (|\mu|^2 t^2+|\zeta|^2 t)}\big]\big) \\
  &=f(x)+2(|\mu|^2+M)\|D^2f\|_\infty t
 \end{align*}
 for all $t\in [0,t_1]$ and $x\in\Rd$. Furthermore, we apply condition~\eqref{eq:mgf2} with 
 \[ C:=\frac{(4|\mu|^2+2M)\|D^2f\|_\infty+|\mu|}{\sqrt{2\alpha\|D^2f\|_\infty}} \] 
 to obtain $R\geq 1$ and $t_2\in (0,1]$ with 
 \begin{align*}
  &\rho[f(x+\zeta_t)-\theta^T\zeta_t] \\
  &=f(x)+(Df(x)-\theta)^T\mu t \\
  &\quad\: +\frac{1}{\alpha}\log\left(\cE\left[\exp\left(\alpha\sqrt{t}(Df(x)-\theta)^T\zeta
    +\alpha\int_0^1\zeta_t^T D^2f(x+s\zeta_t)\zeta_t (1-s)\,\d s\right)\right]\right) \\
  &\geq f(x)+(Df(x)-\theta)^T\mu t- 2 |\mu|^2\|D^2 f\|_\infty t^2
    +\frac{1}{\alpha}\log\big(\cE\big[e^{\sqrt{t}\alpha (Df(x)-\theta)^T\zeta - 2 t \alpha\|D^2 f\|_\infty|\zeta|^2}\big]\big) \\
  &\geq f(x)+(Df(x)-\theta)^T\mu t-2|\mu|^2\|D^2 f\|_\infty t^2+C\sqrt{2\alpha\|D^2f\|_\infty}|Df(x)-\theta|t \\
  &\geq f(x) + |Df(x) - \theta|\big(C\sqrt{2\alpha\|D^2f\|_\infty}-|\mu|\big)t - 2|\mu|^2\|D^2f\|_\infty t\\
  &\geq f(x)+2(|\mu|^2+M)\|D^2f\|_\infty t
 \end{align*}
 for all $t\in [0,t_2]$, $\theta\in\Rd$ with $|Df(x)-\theta|\geq R$ and $x\in\Rd$. Since $Df$
 is bounded, we have therefore verified equation~\eqref{eq:bounded}. Now, one can proceed line by 
 line as in the proof of Theorem~\ref{thm:main} to show that the limit $I'(0)f\in\Cb$
 exists and is given by 
 \[ (I'(0)f)(x)=\inf_{|\theta|\leq R}G_\theta(D^2f(x),Df(x))-\inf_{|\theta|\leq R}G_\theta(0,0)
    \quad\mbox{for all } x\in\Rd, \]
 where $R\geq 0$ satisfies equation~\eqref{eq:bounded}. Since the left-hand side does not depend on
 the particular choice of $R$, we obtain 
 \[ (I'(0)f)(x)=\inf_{\theta\in\Rd}G_\theta(D^2f(x),Df(x))-\inf_{\theta\in\Rd}G_\theta(0,0)
    \quad\mbox{for all } x\in\Rd. \]
 Finally, we observe that the choice of $R\geq 0$ and $t_0>0$ in equation~\eqref{eq:bounded} only 
 depends on $\|Df\|_\infty$ and $\|D^2f\|_\infty$. Hence, again line by line as in the proof of
 Theorem~\ref{thm:main}, one can show that, for every $r\geq 0$, there exist $c\geq 0$ and $t_0>0$
 with
 \[ \|I(t)f-f\|_\infty\leq c\big(\|Df\|_\infty+\|D^2f\|_\infty\big)t \]
 for all $t\in [0,t_0]$ and $f\in\Cb^2$ with $\|f\|_\infty\leq r$ and $0<\|Df\|_\infty+\|D^2f\|_\infty\leq 1$.\
 It follows from~\cite[Corollary~2.14]{BKN23} that condition~(vi) is satisfied.\
 Theorem~\ref{thm:cher} and Theorem~\ref{thm:unique} yield the claim.
\end{proof}

\subsection{Proof of Theorem~\ref{thm:alpha}}
\label{sec:proof4}

This proof of also based on the stability results in~\cite{BKN23} and thus very similar 
to the proof of Theorem~\ref{thm:unbounded}. Recall that, for every $\alpha>0$, we now 
denote by $(S_\alpha(t))_{t\geq 0}$ the semigroup from Theorem~\ref{thm:unbounded} 
previously denoted by $(S(t))_{t\geq 0}$ and by $A_\alpha f$ its generator previously 
denoted by $Af$.

\begin{proof}[Proof of Theorem~\ref{thm:alpha}]
 In order to apply~\cite[Theorem~2.3]{BKN23}, we have to verify~\cite[Assumption~2.2]{BKN23}. 
 The conditions~(i),~(iii) and~(iv) follow from Theorem~\ref{thm:cher} and the corresponding 
 estimates in the proof of Theorem~\ref{thm:main}. In order to verify condition~(v), we show that, 
 for every $\eps>0$, there exists $\alpha_0>0$ with  
 \begin{equation} \label{eq:alpha}
  \left|\inf_{\theta\in\Rd}G_{\alpha,\theta}(D^2 f(x),0)-\frac{1}{2}\cE[\zeta^T D^2f(x)\zeta]\right|\leq\eps
 \end{equation}
 for all $\alpha\geq\alpha_0$, $f\in\Cb^2$ and $x\in\Rd$, where 
 \[ G_{\alpha,\theta}(a,b):=\frac{1}{2}\cE[\zeta^T a\zeta+\alpha |(b-\theta)^T\zeta|^2]+(b-\theta)^T\mu \]
 for all $a\in\R^{d\times d}$ and $b\in\Rd$. Let $\alpha>0$, $f\in\Cb^2$ and $x\in\Rd$. 
 Choosing $\theta=0$ yields 
 \[ \inf_{\theta\in\Rd}G_{\alpha,\theta}(D^2 f(x),0)\leq\frac{1}{2}\cE[\zeta^T D^2f(x)\zeta]. \]
 Furthermore, by condition~\eqref{eq:strict_ell} and Lemma~\ref{lem:E}(ix), there exists $\delta>0$ with
 \begin{align*}
  G_{\alpha,\theta}(D^2 f(x),0) 
  &=\frac{1}{2}\cE\big[\zeta^T D^2f(x)\zeta+\alpha|\theta^T\zeta|^2\big]-\theta^T\mu \\
  &\geq\frac{1}{2}\cE[\zeta^T D^2f(x)\zeta]-\frac{1}{2}\alpha\cE[-|\theta^T\zeta|^2]-\theta^T\mu \\
  &\geq\frac{1}{2}\cE[\zeta^T D^2f(x)\zeta]+\frac{1}{2}\alpha\delta |\theta|^2-\theta^T\mu,
 \end{align*}
 for all $\theta\in\Rd$ with $\theta^T\mu\neq 0$. Hence, for $c_\alpha:=\frac{2|\mu|}{\alpha\delta}$,
 we obtain 
 \[ \inf_{\theta\in\Rd}G_{\alpha,\theta}(D^2 f(x),0)=\inf_{|\theta|\leq c_\alpha}G_{\alpha,\theta}(D^2 f(x),0) \]
 and Lemma~\ref{lem:E}(viii) implies 
 \begin{align*}
  &\left|\inf_{\theta\in\Rd}G_\theta(D^2f(x),0)-\frac{1}{2}\cE[\zeta^T D^2f(x)\zeta]\right| \\
  &\leq\sup_{|\theta|\leq c_\alpha}\left|\frac{1}{2}\cE\big[\zeta^T D^2f(x)\zeta+\alpha|\theta^T\zeta|^2\big] 
    -\theta^T\mu-\frac{1}{2}\cE[\zeta^T D^2f(x)\zeta]\right| \\
  &\leq\sup_{|\theta|\leq c_\alpha}\Big(\frac{1}{2}\cE[\alpha|\theta^T\zeta|^2]+|\theta||\mu|\Big)
    \leq c_\alpha\Big(\frac{|\mu|}{\delta}\cE[|\zeta|^2]+|\mu|\Big).
 \end{align*}
 It follows from $c_\alpha\to 0$ as $\alpha\to\infty$ that inequality~\eqref{eq:alpha} is valid.
 Since Corollary~\ref{cor:Rd} yields 
 \[ (A_\alpha f)(x)=\inf_{\theta\in\Rd}G_{\alpha,\theta}(D^2f(x),0)-\inf_{\theta\in\Rd}G_{\alpha,\theta}(0,0) \]
 and inequality~\eqref{eq:alpha} can also be applied on the constant function $f\equiv 0$, we obtain
 \[ \lim_{\alpha\to\infty}\sup_{x\in\Rd}\left|\frac{1}{2}\cE[\zeta^T D^2f(x)\zeta]-(A_\alpha f)(x)\right|=0
    \quad\mbox{for all } f\in\Cb^2. \]
 This shows that condition~(v) is satisfied. It remains to verify condition~(ii).\ For every risk aversion 
 parameter $\alpha>0$ and volume constraint $R\geq 0$, we denote by $(S_{\alpha,R}(t))_{t\geq 0}$
 the semigroup from Theorem~\ref{thm:main} corresponding to $\Theta:=\Rd$ and by $A_{\alpha,R}f$
 its generator.\ We show that there exists $R_0\geq 0$ with 
 \begin{equation} \label{eq:alpha2}
  \|A_{\alpha,R}f\|_\infty\leq\frac{1}{2}\cE[|\zeta|^2]\cdot\|D^2f\|_\infty
 \end{equation}
 for all $\alpha\geq 1$, $R\geq R_0$ and $f\in\Cb^2$ with $\|Df\|_\infty+\|D^2f\|_\infty\leq 1$.\
 It follows from the proof of Theorem~\ref{thm:unbounded} that there exist $\delta>0$ and $c_1,c_2\geq 0$ with
 \[ G_\theta(D^2f(x),Df(x))\geq\alpha\big(\delta |\theta|^2-c_1|\theta|-c_2\big) \]
 for all $\alpha\geq 1$, $f\in\Cb^2$ with $\|Df\|_\infty+\|D^2f\|_\infty\leq 1$ and $\theta,x\in\Rd$.\
 Hence, there exists $R_0\geq 0$ with $A_{\alpha,R}f=A_\alpha f$ for all $\alpha\geq 1$, $R\geq R_0$ 
 and $f\in\Cb^2$ with $\|Df\|_\infty+\|D^2f\|_\infty\leq 1$. Corollary~\ref{cor:Rd} and Lemma~\ref{lem:E}(viii) 
 now imply that inequality~\eqref{eq:alpha2} is valid. Furthermore, we obtain from~\cite[Corollary~2.16]{BKN23}
 that, for every $T\geq 0$ and $(f_n)_{n\in\N}\subset\Cb$ with $f_n\downarrow 0$, 
 \[ \sup_{\alpha\geq 1}\sup_{R\geq R_0}\sup_{t\in [0,T]}S_{\alpha,R}(t)f_n\downarrow 0. \]
 The uniform continuity from above w.r.t. $\alpha\geq 1$ is preserved in the limit $R\to\infty$, i.e., 
 condition~(ii) is satisfied\footnote{The result in~\cite{BKN23} is only stated for sequences of semigroups 
 but it is actually valid for families $(S_i)_{i\in I}$ of semigroups $(S_i(t))_{t\geq 0}$ which are 
 parameterized by an arbitrary index set. Here, we choose $i:=(\alpha,R)$ and $I:=[1,\infty)\times [R_0,\infty)$.}.
 The first part of the claim follows from~\cite[Theorem~2.5]{BKN23} since Theorem~\ref{thm:unique}
 guarantees that the limit does not depend on the choice of the sequence $\alpha_n\to\infty$. Furthermore,
 as discussed in Subsection~\ref{sec:prel.chernoff}, the family $(T(t))_{t\geq 0}$ is a strongly continuous 
 convex monotone semigroup on $\Cb$ with generator
 \[ (Bf)(x)=\frac{1}{2}\cE[\zeta^T D^2f(x)\zeta] 
    \quad\mbox{for all } f\in\Cb^2 \mbox{ and } x\in\Rd. \]
 Hence, the second part of the claim also follows from Theorem~\ref{thm:unique}.
\end{proof}

\appendix

\section{Basic convexity estimates} 
\label{app:convex}

\begin{lemma} \label{lem:lambda}
 Let $X$ be a vector space and $\Phi\colon X\to\R$ be a convex functional. Then,
 \[ \Phi(x)-\Phi(y)\leq\lambda\left(\Phi\left(\frac{x-y}{\lambda}+y\right)-\Phi(y)\right)
 	\quad\mbox{for all } x,y\in X \mbox{ and } \lambda\in (0,1]. \]
\end{lemma}
\begin{proof}
 For every $x,y\in X$ and $\lambda\in (0,1]$, 
 \begin{align*}
  \Phi(x)-\Phi(y)
  &=\Phi\left(\lambda\left(\frac{x-y}{\lambda}+y\right)+(1-\lambda)y\right)-\Phi(y) \\
  &\leq\lambda\Phi\left(\frac{x-y}{\lambda}+y\right)+(1-\lambda)\Phi(y)-\Phi(y) \\
  &=\lambda\left(\Phi\left(\frac{x-y}{\lambda}+y\right)-\Phi(y)\right). \qedhere
 \end{align*}
\end{proof}

The next lemma states some basic properties of convex and sublinear expectations. 
Convex expectations generalize sublinear expectations by relaxing the conditions~(iii)
and~(iv) in the definition of a sublinear expectation to
\[ \cE[\lambda X+(1-\lambda)Y]\leq\lambda\cE[X]+(1-\lambda)\cE[Y]
    \quad\mbox{for all } \lambda\in [0,1] \mbox{ and } X,Y\in\cH. \]
Furthermore, convex expectations correspond to convex risk measures which are defined 
on losses rather than positions.

\begin{lemma} \label{lem:E}
 For a convex expectation space $(\Omega,\cH,\cE)$ the following statements are valid:
 \begin{enumerate}
  \item $\cE[X+c]=\cE[X]+c$ for all $X\in\cH$ and $c\in\R$.
  \item $|\cE[X]-\cE[Y]|\leq\|X-Y\|_\infty$ for all bounded $X,Y\in\cH$. 
  \item $\cE[\lambda X]\leq\lambda\cE[X]$ for all $\lambda\in [0,1]$ and $X\in\cH$.
  \item $-\cE[-X]\leq\cE[X]$ for all $X\in\cH$. 
  \item $|\cE[X]|\leq\cE[|X|]$ for all $X\in\cH$.
  \item Let $X\in\cH$ with $\cE[aX]=0$ for all $a\in\R$. Then, it holds 
   \[ \cE[X+Y]=\cE[Y] \quad\mbox{for all } Y\in\cH. \]
 \end{enumerate}
 Moreover, if $\cE$ is sublinear, then the following statements are valid:
 \begin{enumerate}[resume]
     \item $\cE[X] - \cE[Y] \le \cE[X - Y]$ for all $X,Y \in \cH$.
     \item $|\cE[X]-\cE[Y]|\leq \cE[|X - Y|]$ for all $X,Y \in \cH$.
     \item $\cE[X + Y] \ge \cE[X] - \cE[-Y]$ for all $X,Y \in \cH$.
 \end{enumerate}
\end{lemma}
\begin{proof}
 The properties (i)-(vi) coincide with~\cite[Lemma B.2]{BK22}.\ Using the sublinearity of~$\cE$, 
 property~(vii) is obtained by rearranging the inequality
 \[ \cE[X]=\cE[X-Y+Y]\leq\cE[X-Y]+\cE[Y]. \]
 Property~(viii) follows from property~(vii) by changing the roles of $X$ and $Y$.\ 
 Finally, the sublinearity of $\cE$ implies $\cE[X]\leq\cE[X+Y]+\cE[-Y]$ showing 
 that property~(ix) is valid.
\end{proof}

\section{Exponential moment estimates}
\label{app:mgf}

\begin{lemma}
 Let $\zeta$ be bounded and write $\cE[\,\cdot\,]=\sup_{\Q\in\cQ}\E_\Q[\,\cdot\,]$.
 Assume that 
 \[ \E_\Q\big[(\theta^T\zeta)^{2k-1}\big]=0 \quad\mbox{for all } \theta\in\Rd \mbox{ and } k\in\N \] 
 and that there exists $\delta>0$ with $\cE[|\theta^T\zeta|^2]\geq\delta|\theta|^2$
 for all $\theta\in\Rd$.\ Then, the conditions~\eqref{eq:mgf1} and~\eqref{eq:mgf2} are satisfied.
\end{lemma}
\begin{proof}
 Let $M:=\sup_{\omega\in\Omega}|\zeta(\omega)|^2$. Regarding condition~\eqref{eq:mgf1}, 
 we observe that 
 \[ \log(\cE[e^{t|\zeta|^2}])\leq Mt \quad\mbox{for all } t\geq 0. \]
 In order to verify condition~\eqref{eq:mgf2}, let $t\geq 0$ and $\theta\in\Rd$.\
 For every $\Q\in\cQ$, the dominated convergence theorem and Jensen's inequality imply
 \begin{align*}
  \cE\big[e^{\sqrt{t}\theta^T\zeta}\big]
  &\geq\E_\Q\big[e^{\sqrt{t}\theta^T\zeta}\big]
  =\sum_{k=0}^\infty\frac{\E_\Q[(\sqrt{t}\theta^T\zeta)^k]}{k!}
  =1+\sum_{k=1}^\infty\frac{\E_\Q[(\sqrt{t}\theta^T\zeta)^{2k}]}{(2k)!} \\
  &\geq 1+\sum_{k=1}^\infty\frac{\E_\Q[(\sqrt{t}\theta^T\zeta)^2]^k}{(2k)!} 
  =\cosh\big(\E_\Q[(\sqrt{t}\theta^T\zeta)^2]^{1/2}\big).
 \end{align*}
 Since hyperbolic cosine is a continuous function, we can choose a sequence $(\Q_n)_{n\in\N}\subset\cQ$ 
 with $\E_{\Q_n}[(\sqrt{t}\theta^T\zeta)^2]\to\cE[(\sqrt{t}\theta^T\zeta)^2]$ to obtain 
 \begin{equation} \label{eq:cosh1}
  \cE\big[e^{\sqrt{t}\theta^T\zeta}\big]
  \geq\cosh\big(\cE[(\sqrt{t}(\theta^T\zeta)^2]^{1/2}\big)
  \geq\cosh(\sqrt{t}\sqrt{\delta}|\theta|).
 \end{equation}

 It remains to show that, for every $C\geq 0$, there exist $t_0>0$ and $R\geq 0$ with
 \begin{equation} \label{eq:cosh2}
  \log(\cosh(\sqrt{t}x))\geq Ctx \quad\mbox{for all } t\in [0,t_0] \mbox{ and } x\geq R.
 \end{equation}
 For every $t\geq 0$, we consider the function $\psi_t\colon\R\to\R,\; x\mapsto\log(\cosh(\sqrt{t}x))-Ctx$.
 It holds
 \[ \psi_t'(x)=\sqrt{t}\tanh(\sqrt{t}x)-Ct\geq\sqrt{t}\tanh(\sqrt{t}2C)-Ct 
    \quad\mbox{for all } x\geq 2C. \]
 Since $\tanh(z)=z+o(z)$ for $z\to 0$, there exists $t_1>0$ with 
 \[ \psi_t'(\sqrt{t}x)\geq\frac{\sqrt{t}\sqrt{t}2C}{2}-Ct=0 
    \quad\mbox{for all } t\in [0,t_1] \mbox{ and } x\geq 2C. \]
 In addition, since $\log(\cosh(z))=\frac{z^2}{2}+o(z^2)$, there exists $t_2>0$ with
 \[ \psi_t(3C)\geq\frac{(3C\sqrt{t})^2}{3}-3C^2t=0 \quad\mbox{for all } t\in [0,t_2]. \]
 This shows that inequality~\eqref{eq:cosh2} is valid with $t_0:=t_1\wedge t_2$ and $R:=3C$.
 Hence, for every $C\geq 0$, it follows from inequality~\eqref{eq:cosh1} and~\eqref{eq:cosh2} 
 that there exists $t_0>0$ and $R\geq 1$ with 
 \[ \log\big(\cE\big[e^{\sqrt{t}\theta^T\zeta-t|\zeta|^2}\big]\big)
    \geq\log\big(e^{-Mt}\cosh(\sqrt{t}\sqrt{\delta}|\theta|)\big)
    \geq C|\theta|t \]
 for all $t\in [0,t_0]$ and $|\theta|\geq R$. This shows that condition~\eqref{eq:mgf2} is valid. 
\end{proof}

\begin{lemma}
 Let $\Lambda$ be a bounded set of positive semi-definite symmetric $d\times d$ matrices and
 define
 \[ \cE[f(\zeta)]:=\sup_{\Sigma\in\Lambda}\int_{\Rd}f(y)\,\cN(0,\Sigma)(\d y) \quad\mbox{for all } f\in\Cb, \]
 where $\cN(0,\Sigma)$ denotes the normal distribution with mean zero and covariance matrix $\Sigma$.
 Then, condition~\eqref{eq:mgf1} is satisfied. Furthermore, if there exists $\delta>0$ with
 \[ \sup_{\Sigma\in\Lambda}\theta^T\Sigma\theta\geq\delta |\theta|^2 \quad\mbox{for all } \theta\in\Rd, \]
 then condition~\eqref{eq:mgf2} is satisfied as well.
\end{lemma}
\begin{proof}
 Denote by $M:=\sup_{\Sigma\in\Lambda}\sup_{|x|=1}x^T\Sigma x$ the largest eigenvalue 
 among all covariance matrices and choose $t_0\in (0,1/(2M))$.\ In particular, for every 
 $t\in [0,t_0]$ and $\Sigma\in\Lambda$, the matrix $\one-2t\Sigma$ is invertible and 
 satisfies $\det(\one-2t\Sigma)\geq (1-2Mt)^d$. We obtain 
 \begin{align*}
  \int_{\Rd}e^{t|x|^2}\,\cN(0,\Sigma)(\d x)
  &=\frac{1}{(2\pi)^{d/2}}\int_{\Rd}e^{t|\sigma x|^2}e^{-\frac{|x|^2}{2}}\,\d x \\
  &=\frac{\det(\tilde{\Sigma})^{1/2}}{(2\pi)^{d/2}\det(\tilde{\Sigma})^{1/2}}
    \int_{\Rd}e^{-\frac{x^T(\tilde{\Sigma}^{-1})x}{2}}\,\d x \\
  &=\det(\tilde{\Sigma})^{1/2}\leq\frac{1}{(1-2Mt)^{d/2}}\leq e^{Mdt}
 \end{align*}
 for all $t\in [0,t_0]$ and $\Sigma\in\Lambda$, where $\tilde{\Sigma}:=(\one-2t\Sigma)^{-1}$
 and $\sigma\in\R^{d\times d}$ is a symmetric matrix with $\sigma^2=\Sigma$. This shows 
 that condition~\eqref{eq:mgf1} is valid. 
 
 In order to verify condition~\eqref{eq:mgf2}, let $\theta\in\Rd$ and $\Sigma\in\Lambda$.
 For every $t\in [0,t_0]$, we use that $\one+2t\Sigma$ is invertible and 
 $\det(\one+2t\Sigma)\leq (1+2Mt)^d$ to obtain 
 \begin{align*}
  \cE\big[e^{\sqrt{t}\theta^T\zeta-t|\zeta|^2}\big] 
  &\geq\frac{1}{(2\pi)^{d/2}}\int_{\Rd}e^{\sqrt{t}\theta^T\sigma x-t|\sigma x|^2}e^{-\frac{|x|^2}{2}}\,\d x \\
  &=\frac{\det(\tilde{\Sigma})^{1/2}}{(2\pi)^{d/2}\det(\tilde{\Sigma}))^{1/2}}
    \int_{\Rd}e^{\sqrt{t}\theta^T\sigma x} e^{-\frac{1}{2}x^T\tilde{\Sigma}^{-1} x}\,\d x \\
  &=\det(\tilde{\Sigma}^{-1})^{-1/2}e^{\frac{1}{2}t (\sigma\theta)^T\tilde{\Sigma}\sigma\theta} \\
  &\geq (1+Mt)^{-d/2}e^{\frac{1}{2}t (\sigma\theta)^T\tilde{\Sigma}\sigma\theta}
    \geq e^{-Mdt}e^{\frac{1}{2}t (\sigma\theta)^T\tilde{\Sigma}\sigma\theta},
 \end{align*}
 where $\tilde{\Sigma}:=(\one+2t\Sigma)^{-1}$ and $\sigma\in\R^{d\times d}$ is a symmetric matrix with
 $\sigma^2=\Sigma$. Moreover,
 \[ \sup_{\Sigma\in\Lambda}(\sigma\theta)^T\tilde{\Sigma}\sigma\theta
    \geq\sup_{\Sigma\in\Lambda}\frac{|\sigma^T \theta|^2}{1+2Mt}
    =\sup_{\Sigma\in\Lambda}\frac{\theta^T\Sigma\theta}{1+2Mt}\geq\frac{\delta |\theta|^2}{1+2Mt_0}
    \quad\mbox{for all } t\in [0,t_0]. \]
 Hence, for every $C\geq 0$, there exists $R\geq 0$ with
 \[ \log\big(\cE\big[e^{\sqrt{t}\theta^T\zeta-t|\zeta|^2}\big]\big)
    \geq\frac{\delta |\theta|^2 t}{2+4Mt_0}-Mdt\geq C|\theta|t \]
 for all $t\in [0,t_0]$ and $|\theta|\geq R$. This shows that condition~\eqref{eq:mgf2} is satisfied. 
\end{proof}

\bibliographystyle{abbrv}

\end{document}